\documentclass[a4paper,UKenglish,cleveref, autoref, thm-restate,authorcolumns]{lipics-v2019}

\usepackage{comment}
\usepackage{todonotes}
\usepackage{tikz}
\usetikzlibrary{positioning,shapes,chains,backgrounds,arrows,chains,fit,snakes,calc,automata}

\newcommand{\tikzmark}[1]{\tikz[overlay,remember picture] \node (#1) {};}
\newcommand{\DrawBox}[3][]{%
    \tikz[overlay,remember picture]{
    \draw[black,#1]
      ($(#2)+(-0.75em,3.0ex)$) rectangle
      ($(#3)+(0.0em,-0.75ex)$);}
}

\newlength\mylen
\newcolumntype{C}{>{\hfil$}p{\mylen}<{$\hfil}} %

\title{Computing the original eBWT faster, simpler, and with less memory} %

\titlerunning{Computing the original eBWT faster, simpler, and with less memory} %

\author{Christina Boucher}{Department of Computer and Information Science and Engineering, University of Florida, Gainesville, FL, United States}{c.boucher@cise.ufl.edu}{https://orcid.org/0000-0001-9509-9725}{}

\author{Davide Cenzato}{Department of Computer Science, University of Verona, Verona, Italy}{davide.cenzato@univr.it}{https://orcid.org/0000-0002-0098-3620}{}

\author{{Zs}uzsanna Lipt{\'a}k}{Department of Computer Science, University of Verona, Verona, Italy}{zsuzsanna.liptak@univr.it}{https://orcid.org/0000-0002-3233-0691}{}

\author{Massimiliano Rossi}{Department of Computer and Information Science and Engineering, University of Florida, Gainesville, FL, United States}{rossi.m@ufl.edu}{https://orcid.org/0000-0002-3012-1394}{}%

\author{Marinella Sciortino}{Department of Computer Science, University of Palermo, Palermo, Italy}{marinella.sciortino@unipa.it}{https://orcid.org/0000-0001-6928-0168}{}

\authorrunning{C.\ Boucher, D.\ Cenzato, Zs.\ Lipt{\'a}k, M.\ Rossi, M.\ Sciortino} %

\Copyright{Christina Boucher, Davide Cenzato, Zsuzsanna Lipt{\'a}k, Massimiliano Rossi, Marinella Sciortino} %



\relatedversion{} %

\nolinenumbers %

\hideLIPIcs  %

\EventEditors{John Q. Open and Joan R. Access}
\EventNoEds{2}
\EventLongTitle{42nd Conference on Very Important Topics (CVIT 2016)}
\EventShortTitle{CVIT 2016}
\EventAcronym{CVIT}
\EventYear{2016}
\EventDate{December 24--27, 2016}
\EventLocation{Little Whinging, United Kingdom}
\EventLogo{}
\SeriesVolume{42}
\ArticleNo{23}

\usepackage[utf8]{inputenc}
\usepackage[group-separator={,}]{siunitx}
\usepackage[useregional]{datetime2}
\usepackage{framed}
\usepackage{algorithm2e}

\usepackage[noend]{algpseudocode}

\newcommand{\BWT}{\ensuremath{\mathrm{BWT}}}

\newcommand{\LCP}{\ensuremath{\mathrm{LCP}}}
\newcommand{\PFP}{\ensuremath{\mathrm{PFP}}}
\newcommand{\SA}{\ensuremath{\mathrm{SA}}}

\newcommand{\CA}{\ensuremath{\mathrm{CA}}}

\newcommand{\GCA}{\ensuremath{\mathrm{GCA}}}

\def\eBWT{\ensuremath{\mathrm{eBWT}}}
\newcommand{\BBWT}{\ensuremath{\mathrm{BBWT}}}
 
\def\LMS{\textit{LMS}}

\newcommand{\lcp}{\texttt{lcp}}

\renewcommand{\root}{\ensuremath{\mathrm{root}}}
\renewcommand{\exp}{\ensuremath{\mathrm{exp}}}
\newcommand{\lex}{\ensuremath{\mathrm{lex}}}

\def\chr19{\texttt{chr19}}
\def\salmonella{\texttt{salmonella}}
\def\sars{\texttt{sars-cov2}}
\def\ours{\texttt{pfpebwt}}
\def\rope{\texttt{ropebwt2}}
\def\gsufsort{\texttt{gsufsort}}
\def\egap{\texttt{egap}}

\newcommand*{\PlusPlus}{%
\kern0.3ex\raisebox{-0ex}{\scalebox{0.8}{\kern-0.4ex+}}%
\kern-0ex\raisebox{0.5ex}{\scalebox{0.8}{\kern-0.4ex+}}}

\newcommand{\conj}{\textrm{conj}}

\begin{document}

\maketitle

\begin{abstract}

Given an input string, the {\em Burrows-Wheeler Transform} (\BWT) can be seen as a reversible permutation of it that allows efficient compression and fast substring queries.  Due to these properties, it has been widely applied in the analysis of genomic sequence data, enabling important tasks such as read alignment. Mantaci et al.\ [TCS2007] extended the notion of the \BWT{} to a collection of strings by defining the {\em extended Burrows-Wheeler Transform} (\eBWT).  This definition requires no modification of the input collection, and has the property that the output is independent of the order of the strings in the collection. However, over the years, the term \eBWT\ has been used more generally to describe any \BWT\ of a collection of strings.  The fundamental property of the original definition (i.e., the independence from the input order) is frequently disregarded.
In this paper, we propose a simple linear-time algorithm for the construction of the original \eBWT, which does not require the preprocessing of Bannai et al. [CPM 2021]. As a byproduct, we obtain the first linear-time algorithm for computing the \BWT\ of a single string that uses neither an end-of-string symbol nor Lyndon rotations. 

We also combine our new \eBWT{} construction with a variation of {\em prefix-free parsing} (PFP) [WABI 2019] to allow for construction of the \eBWT\ on large collections of genomic sequences. We implement this combined algorithm (\ours{}) and evaluate it on a collection of human chromosomes 19 from the 1,000 Genomes Project, on a collection of Salmonella genomes from GenomeTrakr, and on a collection of SARS-CoV2 genomes from EBI's COVID-19 data portal. We demonstrate that \ours{} is the fastest method for all collections, with a maximum speedup of 7.6x on the second best method. The peak memory is at most 2x larger than the second best method. Comparing with methods that are also, as our algorithm, able to report suffix array samples,  we obtain a 57.1x improvement in peak memory. The source code is publicly available at \url{https://github.com/davidecenzato/PFP-eBWT}.

\end{abstract}
\textbf{}

\section{Introduction}

In the last several decades, the number of sequenced human genomes has been growing at unprecedented pace. In 2015 the number of sequenced genomes was doubling every $7$ months~\cite{stephens_big_2015} -- a pace that has not slowed into the current decade.  The plethora of resulting sequencing data has expanded our knowledge of the biomarkers responsible for human disease and phenotypes \cite{Berner,100K,1000genomes}, the evolutionary history between and among species \cite{VGP}, and will eventually help realize the personalization of healthcare \cite{personal}.  However, the amount of data for any individual species is large enough that it poses challenges with respect to storage and analysis.  One of the most well-known and widely-used methods for compressing and indexing data that has been applied in bioinformatics is the Burrows-Wheeler Transform (\BWT), which is a text transformation that compresses the input in a manner that also allows for efficient substring queries.  Not only can it be constructed in linear-time in the length of the input, it is also reversible -- meaning the original input can be constructed from its compressed form. The \BWT{} is formally defined over a single input string; thus, in order to define and construct it for one or more strings, the input strings need to be concatenated or modified in some way.  In  2007 Mantaci et al.~\cite{MantaciRRS07} presented a formal definition of 
the \BWT{} for a multiset of strings, which they call the {\em extended Burrows-Wheeler Transform} (\eBWT).  It is a bijective transformation that sorts the cyclic rotations of the strings of the multiset according to the $\omega$-order relation, an order, defined by considering infinite iterations of each string, which is different from the lexicographic order.

Since its introduction several algorithms have been developed that construct the \BWT\ of collection of strings for various types of biological data including short sequence reads  \cite{DBLP:journals/tcs/BonizzoniVPPR21,DBLP:journals/tcs/BauerCR13,DBLP:journals/bioinformatics/CoxBJR12,egidi2019external,louza2020gsufsort,diaz2021efficient,egidi2019external,Ander2013,GuerriniRosone_Alcob2019,PrezzaPSR19,PrezzaPSR20}, protein sequences \cite{YANG2010742}, metagenomic data \cite{meta} and longer DNA sequences such as long sequence reads and whole chromosomes \cite{rope}.   However, we note that in the development of some of these methods the underlying definition of \eBWT\ was loosened.  For example, \rope~\cite{rope} tackles a similar problem of building what they describe as the FM-index for a multiset of long sequence reads, however, they do not construct the suffix array (\SA) or \SA\ samples, and also, require that the sequences are delimited by separator symbols. Similarly, \gsufsort~\cite{louza2020gsufsort} and \egap~\cite{egidi2019external} construct the \BWT\ for a collection of strings but do not construct the \eBWT\ according to its original definition. \gsufsort~\cite{louza2020gsufsort} requires the collection of strings to be concatenated in a manner that the strings are deliminated by separator symbols that have an augmented relative order among them. \egap~\cite{egidi2019external}, which was developed to construct the \BWT\ and \LCP\ for a collection of strings in external memory, uses the \texttt{gSACA-K} algorithm to construct the suffix array of the concatenated input using an additional $O(\alpha + 1) \log n$ bits, and then constructs the \BWT\ for the collection from the resulting suffix array.  Lastly, we note that there exists a number of methods for construction of the \BWT\ for a collection of short sequence reads, including {\tt ble}~\cite{DBLP:journals/tcs/BonizzoniVPPR21}, {\tt BCR}~\cite{DBLP:journals/tcs/BauerCR13},  {\tt G2BWT}~\cite{diaz2021efficient}, {\tt egsa}~\cite{egsa}; however, these methods make implicit or explicit use of end-of-string symbols appended to strings in the collection. For an example of the effects of these manipulations, see Section~\ref{sec:preliminaries}, and~\cite{CL21} for a more detailed study.

 We present an efficient algorithm for constructing the \eBWT\ that preserves the original definition of Mantaci et al.~\cite{MantaciRRS07}---thus, it does not impose any ordering of the input strings or delimiter symbols.  It is an adaptation of the well-known Suffix Array Induced Sorting (SAIS) algorithm of Nong et al.~\cite{NongZC2011}, which computes the suffix array of a single string $T$ ending with an end-of-string character $\$$. Our adaptation is similar to the algorithm proposed by Bannai et al.~\cite{BannaiKKP21} for computing the \BBWT, which can also be used for computing the \eBWT, after linear-time preprocessing of the input strings.  The key change in our approach is based on the insight that the properties necessary for applying Induced Sorting are valid also for the $\omega$-order between different strings.  As a result, is it not necessary that the input be Lyndon words, or that their relative order be known at the beginning. Furthermore, our algorithmic strategy, when applied to a single string, provides the first linear-time algorithm for computing the \BWT\ of the string that uses neither an end-of-string symbol nor Lyndon rotations. 

We then combine our new \eBWT{} construction with a variation of a preprocessing technique called {\em prefix free parsing} (\PFP). \PFP\ was introduced by Boucher et  al.~\cite{DBLP:journals/almob/BoucherGKLMM19} for building the (run length encoded) \BWT\ of large and highly repetitive input text. Since its original introduction, it has been extended to construct the $r$-index~\cite{recomb19}, been applied as a preprocessing step for building grammars~\cite{BigRePair}, and used as a data structure itself~\cite{boucher2020pfp}.  Briefly, \PFP\ is a one-pass algorithm that divides the input into overlapping variable length phrases with delimiting prefixes and suffixes; which in effect, leads to the construction of what is referred to as the dictionary and parse of the input.  It follows that the \BWT\ can be constructed in the space that is proportional to the size of the dictionary and parse, which is expected to be significantly smaller than linear for repetitive text.

In our approach, prefix-free parsing is applied to obtain a parse that is a multiset of cyclic strings ({\em cyclic prefix-free parse}) on which our \eBWT\ construction is applied. We implement our approach (called \ours), measure the time and memory required to build the \eBWT\ for sets of increasing size of chromosome 19, {\it Salmonella}, and SARS-CoV2 genomes, and compare this to that required by \gsufsort, \rope, and \egap.  We show that $\ours$ is consistently faster and uses less memory than \gsufsort\ and \egap\ on reasonably large input ($\geq$ 4 copies of chromosome 19, $\geq$ 50 {\it Salmonella} genomes, and $\geq$ 25,000 SARS-CoV2 genomes).  Although \rope\ uses less memory than $\ours$ on large input, \ours\ is 7x more efficient in terms of wall clock time, and 2.8x in terms of CPU time. Moreover, \ours\ is capable of reporting \SA\ samples in addition to the \eBWT\ with a negligible increase in time and memory~\cite{recomb19}, whereas \rope\ does not have that ability. If we compare \ours\ only with methods that are able to report \SA\ samples in addition to the \eBWT\ (e.g., \egap\ and \gsufsort),  we obtain a 57.1x improvement in peak memory.

\section{Preliminaries}\label{sec:preliminaries}

A string $T=T[1..n]$ is a sequence of characters $T[1]\cdots T[n]$ drawn from an ordered alphabet $\Sigma$ of size $\sigma$. We denote by $|T|$ the length $n$ of $T$, and by $\varepsilon$ the empty string, the only string of length $0$. Given two integers $1 \leq i ,j \leq n$, we denote by $T[i..j]$ the string $T[i]\cdots T[j]$, if $i \leq j$, while $T[i..j]=\varepsilon$ if $i > j$. We refer to $T[i..j]$ as a {\em substring} (or {\em factor}) of $T$, to $T[1..j]$ as the $j$-th {\em prefix} of $T$, and to $T[i..n] = T[i..]$ as the $i$-th {\em suffix} of $T$. A substring $S$ of $T$ is called {\em proper} if $T\neq S$. Given two strings $S$ and $T$, we denote by $\lcp(S,T)$ the length of the {\em longest common prefix} of $S$ and $T$, i.e., $\lcp(S,T) = \max\{i \mid S[1..i] = T[1..i])$. 

Given a string $T=T[1..n]$ and an integer $k$, we denote by $T^k$ the $kn$-length 
string $TT\cdots T$ ($k$-fold concatenation of $T$), and by $T^\omega$ the infinite string $TT\cdots$ obtained by concatenating an infinite number of copies of $T$. 
A string $T$ is called {\em primitive} if $T=S^k$ implies $T=S$ and $k=1$. For any string $T$, there exists a unique primitive word $S$ and a unique integer $k$ such that $T = S^k$. We refer to $S=S[1..\frac nk]$ as $\root(T)$ and to $k$ as $\exp(T)$. Thus, $T = \root(T)^{\exp(T)}.$

We denote by $<_\lex$ the lexicographic order: for two strings $S[1..n]$ and $T[1..m]$, $S <_\lex T$ if $S$ is a proper prefix of $T$, or there exists an index $1 \leq i \leq n,m$ such that $S[1..i-1]=T[1..i-1]$ and $S[i]< T[i]$. Given a string $T[1..n]$, the {\em suffix array}~\cite{mm1993}, denoted by $\SA=\SA_T$, is the permutation of $\{1,\ldots,n\}$ such that $T[\SA[i]..]$ is the $i$-th lexicographically smallest suffix of $T$.

We denote by $\prec_\omega$ the $\omega$-order~\cite{GeRe93,MantaciRRS07}, defined as follows: 
for two strings $S$ and $T$, $S \prec_\omega T$ if $\root(S) = \root(T)$ and $ \exp(S) < \exp(T)$, or 
$S^\omega <_\lex T^\omega$ (this implies $\root(S) \neq \root(T)$). One can verify that the $\omega$-order relation is different from the lexicographic one. For instance,  $CG <_\lex CGA$ but $CGA\prec_\omega CG$. 

The string $S$ is a {\em conjugate} of the string $T$ if $S= T[i..n]T[1..i-1]$, for some $i\in \{1,\ldots, n\}$ (also called the {\em $i$-th  rotation} of $T$). The conjugate $S$ is also denoted $\conj_i(T)$.  It is easy to see that $T$ is primitive if and only if it has $n$ distinct conjugates. A {\em Lyndon word} is a primitive string which is lexicographically smaller than all of its conjugates. 
For a string $T$, the {\em conjugate array}\footnote{Our conjugate array $\CA$ is called {\em circular suffix array} and denoted $\SA_{\circ}$ in~\cite{HonKLST12, BannaiKKP21}, and {\em BW-array} in~\cite{KucherovTV13, EnCombWords}, but in both cases defined for primitive strings only.} $\CA=\CA_T$ of $T$ is the permutation of $\{1,\ldots,n\}$ such that $\CA[i]=j$ if $\conj_j(T)$ is the $i$-th conjugate of $T$ with respect to the lexicographic order, with ties broken according to string order, i.e.\ if $\CA[i] = j$ and $\CA[i'] = j'$ for some $i<i'$, then either $\conj_j(T) <_\lex \conj_{j'}(T)$, or $\conj_j(T) = \conj_{j'}(T)$ and $j<j'$.  Note that if $T$ is a Lyndon word, then $\CA[i]=\SA[i]$ for all $1\leq i\leq n$ \cite{GIA07}. 

Given a string $T$, $U$ a {\em circular} or {\em cyclic substring} of $T$ if it is a factor of $TT$ of length at most $|T|$, or equivalently, if it is the prefix of some conjugate of $T$. For instance, \textit{ATA} is a cyclic substring of \textit{AGCAT}. It is sometimes also convenient to regard a given string $T[1..n]$ itself as {\em circular} (or {\em cyclic}); in this case we set $T[0] = T[n]$ and $T[n+1]=T[1]$. 

\subsection{Burrows-Wheeler-Transform} 
 
The {\em Burrows-Wheeler Transform}~\cite{BW94} of $T$, denoted $\BWT$, is a reversible transformation extensively used in data compression. Given a string $T$, $\BWT(T)$ is a permutation of the letters of $T$ which equals the last column of the matrix of the lexicographically sorted conjugates of $T$. %
The mapping $T \mapsto \BWT(T)$ is reversible, up to rotation. It can be made uniquely reversible by adding to $\BWT(T)$ and index indicating the rank of $T$ in the lexicographic order of all of its conjugates. Given $\BWT(T)$ and an index $i$, the original string $T$ can be computed in linear time~\cite{BW94}. The $\BWT$ itself can be computed from the conjugate array, since for all $i=1,\ldots,n$, $\BWT(T)[i] = T[\CA[i] - 1]$, where $T$ is considered to be cyclic.  

It should be noted that in many applications, it is assumed that an end-of-string-character  (usually denoted $\$$), which is not element of $\Sigma$, is appended to the string; this character is assumed to be smaller than all characters from $\Sigma$. Since $T\$$ has exactly one occurrence of $\$$, $\BWT(T\$)$ is now uniquely reversible, without the need for the additional index $i$, since $T\$$ is the unique conjugate ending in $\$$. Moreover, adding a final $\$$ makes the string primitive, and $\$T$ is a Lyndon word. Therefore, computing the conjugate array becomes equivalent to computing the suffix array, since $\CA_{T\$}[i]=\SA_{T\$}[i]$. Thus, applying one of the linear-time suffix-array computation algorithms~\cite{gonzalo-book} leads to linear-time computation of the $\BWT$. 

When no $\$$-character is appended to the string, the situation is slightly more complex. For primitive strings $T$, first the Lyndon conjugate of $T$ has to be computed (in linear time, \cite{Shiloach81}) and then a linear-time suffix array algorithm can be employed~\cite{GIA07}. For strings $T$ which are not primitive, one can take advantage of the following well-known property of the $\BWT$: let $T = S^k$ and $\BWT(S) = U[1..m]$, then $\BWT(T) = U[1]^kU[2]^k\cdots U[m]^k$~(Prop.~2 in \cite{MantaciRS03}). Thus, it suffices to compute the $\BWT$ of $\root(T)$. The root of $T$ can be found by computing the border array $\bf b$ of $T$: $T$ is a power if and only if ${n}/(n-{\bf b}[n])$ is an integer, which is then also the length of $\root(T)$. The border array can be computed, for example, by the preprocessing phase of  the KMP-algorithm for pattern matching~\cite{KMP77}, in linear time in the length of $T$.

\subsection{Generalized Conjugate Array and Extended Burrows-Wheeler Transform}

Given a multiset of strings ${\cal M} = \{T_1[1..n_1],\ldots,T_m[1..n_m]\}$, the {\em generalized conjugate array} of ${\cal M}$, denoted by $\GCA_{\cal M}$ or just by $\GCA$, contains the list of the conjugates of all strings in ${\cal M}$, sorted according to the $\omega$-order relation. More formally,  $\GCA[i]=(j,d)$ if $\conj_j(T_d)$ is the $i$-th string in the $\preceq_\omega$-sorted list of the conjugates of all strings of ${\cal M}$, with ties broken first w.r.t.\ the index of the string (in case of identical strings), and then w.r.t.\ the index in the string itself. 

The {\em extended Burrows-Wheeler Transform} (\eBWT) is an extension of the \BWT\ to a multiset of strings \cite{MantaciRRS07}. It is a bijective transformation that, given a multiset of strings ${\cal M} = \{T_1,\ldots,T_m\}$, produces a permutation of the characters on the strings in the multiset ${\cal M}$. Formally, $\eBWT({\cal M})$ can be computed by sorting all the conjugates of the strings in the multiset according to the $ \preceq_\omega$-order, and the output is the string obtained by concatenating the last character of each conjugate in the sorted list, together with the set of indices representing the positions of the original strings of $\mathcal{M}$ in the list. Similarly to the \BWT, the \eBWT\ is thus uniquely reversible. The $\eBWT(\mathcal{M})$ can be computed from the generalized conjugate array of $\mathcal{M}$ in linear time, since $\eBWT(\mathcal{M})[i]=T_d[j-1]$ if $\GCA[i]=(j,d)$, where again, the strings in $\mathcal{M}$ are considered to be cyclic. It is easy to see that when $\mathcal{M}$ consists of only one string, i.e.\  $\mathcal{M}=\{T\}$, then $\eBWT(\mathcal{M})=\BWT(T)$.

\begin{example}\label{ex:ex1}
Let ${\cal M} = \{ \textit{GTACAACG}, \textit{CGGCACACACGT}, \textit{C} \}$. Then $\GCA({\cal M})$ is as follows, where we give the pair $(j,d)$ vertically, i.e.\ the first row contains the position in the string, and the second row the index of the string: \\

$
\begin{array}{*{21}{r}}
5 & 3 & 5 & 7 & 6  & 9  & 4  & 4 & 6 & 8 & 1 & 1 & 7 & 10 & 3  & 2 & 8 & 1 & 11 & 2 & 12 \\
1 & 1 & 2 & 2 & 1  & 2  & 1  & 2 & 2 & 2 & 3 & 2 & 1 & 2  & 2  & 2 & 1 & 1 &  2 & 1  & 2 \\ 
\end{array}
$
\medskip 

From the \GCA\ we can compute $\eBWT({\cal M}) = \textit{CTCCACAGAACTAAGCCGCGG}$, with index set $\{11,12,18\}$. Note that e.g.\ the conjugate $\conj_8(T_2)$ comes before $\conj_1(T_3)$, since $\textit{CACGTCGGCACA} \prec_{\omega} C$, because $(CACGTCGGCACA)^{\omega} <_\lex C^{\omega} = CCCC\ldots$ holds. The full list of conjugates is in Appendix~\ref{app:exampleeBWT}. 
\end{example} 

\begin{remark}
Note that if end-of-string symbols are appended to the string of the collection the output of $\eBWT$ could be quite different. For instance, if $\mathcal{M}=\{\textit{GTACAACG}\$_1,$ $\textit{CGGCACACACGT}\$_2, C\$_3\}$, $\eBWT(\mathcal{M})=GTCCTCCAC\$_3AGAAA\$_2ACGCC\$_1GG$. 
\end{remark}

Note that while in the original definition of \eBWT~\cite{MantaciRRS07}, the multiset $\mathcal M$ was assumed to contain only primitive strings, our definition is more general and allows also for non-primitive strings. For example, $\eBWT(\{\textit{ATA, TATA}\}) = {\bf TATTAAA}$, with index set $\{2,6\}$, while $\eBWT(\{\textit{ATA,TA,TA}\}) = {\bf TATTAAA}$, with index set $\{2,6,7\}$. Also the linear-time algorithm for recovering the original multiset can be straightforwardly extended.

The following lemma shows how to construct the generalized conjugate array $\GCA_{\mathcal{M}}$ of a multiset $\mathcal{M}$ of strings (not necessarily primitive), once we know the generalized conjugate array $\GCA_\mathcal{R}$ of the multiset $\mathcal{R}$ of the roots of the strings in $\mathcal{M}$. It follows straightforwardly from the fact that equal conjugates will end up consecutively in the \GCA. 

\begin{lemma}\label{le:ebwt_roots} 
Let $\mathcal{M}=\{T_1,\ldots,T_m\}$ be a multiset of strings and let $\mathcal{R}$ the multiset of the roots of the strings in $\mathcal{M}$, i.e. $\mathcal{R}=\{S_1,\ldots,S_m\}$, where $T_i=(S_i^{r_i})$, with $r_i\geq 1$ for $1\leq i\leq m$. Let $\GCA_{\mathcal R}[1..K] = [(j_1,i_1), (j_2,i_2), \ldots, (j_K,i_K)]$, where $K=\sum_{i=1}^m|S_i|$. The generalized conjugate array is then given by 
\begin{align*}
\GCA_{\mathcal M}[1..N] = [ &(j_1,i_1), (j_1+|S_{i_1}|, i_1), \ldots, (j_1+(r_{i_1}-1)\cdot |S_{i_1}|, i_1), \\
& (j_2,i_2), (j_2+|S_{i_2}|, i_2), \ldots, (j_2 + (r_{i_2}-1)\cdot |S_{i_2}|, i_2), \\
&\ldots \\
& (j_K,i_K), (j_K+|S_{i_K}|, i_K), \ldots, (j_K + (r_{i_K}-1)\cdot |S_{i_K}|, i_K)], 
\end{align*}

\noindent with $N=\sum_{i=1}^m|S_i|\cdot r_i$. 

\end{lemma}

From now on we will assume that the multiset $\mathcal{M}=\{T_1,\ldots,T_m\}$ consists of $m$ primitive strings.

\section{A simpler algorithm for computing the \eBWT \ and \GCA}\label{sec:algo_ebwt}

In this section, we describe our algorithm to compute the \eBWT\ of a multiset of strings ${\cal M}$. We will assume that all strings in ${\cal M}$ are primitive, since we can use Lemma~\ref{le:ebwt_roots} to compute the \eBWT\ of ${\cal M}$ otherwise. Our algorithm is an adaptation of the well-known SAIS algorithm of Nong et al.~\cite{NongZC2011}, which computes the suffix array of a single string $T$ ending with an end-of-string character $\$$. Our adaptation is similar to that of Bannai et al.~\cite{BannaiKKP21} for computing the \BBWT, which can also be used for computing the \eBWT. Even though our algorithm does not improve the latter asymptotically (both are linear time), it is significantly simpler, since it does not require first computing and sorting the Lyndon rotations of the input strings. 

In the following, we assume some familiarity with the SAIS algorithm, focusing on the differences between our algorithm and the original SAIS. Detailed explanations of SAIS can be found in the original paper~\cite{NongZC2011}, or in the books~\cite{ohlebusch-book, louza-book}. 

The main differences between our algorithm and the original SAIS algorithm are: (1) we are comparing conjugates rather than suffixes, (2) we have a multiset of strings rather than just one string, (3) the comparison is done w.r.t.\ the omega-order rather than the lexicographic order, and (4) the strings are not terminated by an end-of-string symbol. 

\medskip

We need the following definition, which is the cyclic version of the definition in~\cite{NongZC2011} (where $S$ stands for smaller, $L$ for larger, and \LMS\ for leftmost-S): 

\begin{definition}[Cyclic types, \LMS-substrings]
Let $T$ be a primitive string of length at least $2$, and $1\leq i \leq |T|$. Position $i$ of $T$ is called {\em (cyclic) S-type} if $\conj_i(T) <_{\lex} \conj_{i+1}(T)$, and {\em (cyclic) L-type} if $\conj_i(T) >_{\lex} \conj_{i+1}(T)$. An S-type position $i$ is called {\em (cyclic) LMS} if  $i-1$ is L-type (where we view $T$ as a cyclic string). An {\em LMS-substring} is a cyclic substring $T[i,j]$ of $T$ such that both $i$ and $j$ are \LMS-positions, but there is no \LMS-position between $i$ and $j$. Given a conjugate $\conj_i(T)$, its {\em LMS-prefix} is the cyclic substring from $i$ to the first \LMS-position strictly greater than $i$ (viewed cyclically). 
\end{definition} 

Since $T$ is primitive, no two conjugates are equal, and in particular, no two adjacent conjugates are equal. Therefore, the type of every position of $T$ is defined.

\begin{example} Continuing Example~\ref{ex:ex1}, \\

$
\begin{array}{*{22}{r}}
G & T & A & C & A & A & C & G & &  & C & G & G & C & A & C & A & C & A & C & G & T \\
S & L & S & L & S  & S  & S  & S & &   & S & L & L & L & S & L  & S & L & S & S & S & L \\
&& * && * &&&&&& * &&&& * && *  && * &&& \\ 
\end{array}
$
\medskip 

\noindent where we mark \LMS-positions with a $*$. The \LMS-substrings are \textit{ACA}, \textit{AACGGTA}, \textit{CGGCA}, and \textit{ACGTC}. The  \LMS-prefix of the conjugate $\conj_7(T_1) = \textit{CGGTACAA}$ is \textit{CGGTA}. 
\end{example} 

\begin{lemma}[Cyclic type properties]~\label{lemma:types} 
Let $T$ be primitive string of length at least $2$. Let $a_1$ be the smallest and $a_{\sigma}$ the largest character of the alphabet. Then the following hold, where $T$ is viewed cyclically: 

\begin{enumerate} 
\item if $T[i] < T[i+1]$, then $i$ is of type $S$, and if $T[i] > T[i+1]$, then $i$ is of type $L$, 
\item if $T[i] = T[i+1]$, then the type of $i$ is the same as the type of $i+1$, 
\item $i$ is of type $S$ iff $T[i'] > T[i]$, where $i' = \min\{ j \mid T[j] \neq T[i]\}$, %
\item if $T[i] = a_1$, then $i$ is of type $S$, and if $T[i] = a_{\sigma}$, then $i$ is of type $L$. 
\end{enumerate} 

\end{lemma}

\begin{proof} 
1. follows from the fact that for all $b,c\in \Sigma$, if $b<c$ then for all $U,V\in \Sigma^*$, $bU \prec_{\omega} cV$; 2. follows by induction from the fact that for all $U,V\in \Sigma^*$, if $U\prec_{\omega} V$, then $cU \prec_{\omega} cV$;  3. and 4. follow from 2. by induction. \end{proof}

\begin{corollary}[Linear-time cyclic type assignment]
Let $T$ be a primitive string of length at least $2$. Then all positions can be assigned a type in altogether at most $2|T|$ steps. 
\end{corollary}

\begin{proof} 
Once the type of one position is known, then the assignment can be done in one cyclic pass over $T$ from right to left, by Lemma~\ref{lemma:types}.  Therefore, it suffices to find the type of one single position. Any position of character $a_1$ or of character $a_{\sigma}$ will do; alternatively, any position $i$ such that $T[i+1] \neq T[i]$, again by Lemma~\ref{lemma:types}. Since $T$ is primitive and has length at least $2$, the latter must exist and can be found in at most one pass over $T$. 
\end{proof}

Let $N$ be the total length of the strings in ${\cal M}$. The algorithm constructs an initially empty array $A$ of size $N$, which, at termination, will contain the \GCA\ of ${\cal M}$. The algorithm also returns the set $\mathcal{I}$ containing the set of indices in $A$ representing the positions of the strings of $\mathcal{M}$. The overall procedure consists of the following steps: 

\begin{quote} Algorithm {\tt SAIS-for-\eBWT}

\begin{itemize}
\item[Step 1] remove strings of length $1$ from ${\cal M}$ (these will be added back at the end) 
\item[Step 2] assign cyclic types to all positions of strings from ${\cal M}$ 
\item[Step 3] use procedure {\tt Induced Sorting} to sort cyclic \LMS-substrings 
\item[Step 4] assign names to cyclic \LMS-substrings; if all distinct, go to Step 6 
\item[Step 5] recurse on new string multiset ${\cal M'}$, returning array $A'$, map $A'$ back to $A$ 
\item[Step 6] use procedure {\tt Induced Sorting} to sort all positions in ${\cal M}$, add length-1 strings in their respective positions, return $(A,\mathcal{I})$
\end{itemize}

\end{quote}

At the heart of the algorithm is the procedure {\tt Induced Sorting} of~\cite{NongZC2011} (Algorithms 3.3 and 3.4), which is used once to sort the \LMS-substrings (Step 3), and once to induce the order of all conjugates from the correct order of the \LMS-positions (Step 6), as in the original SAIS. Before  sketching this procedure, we need to define the order according to which the \LMS-substrings are sorted in Step 2. Note that our definition of \LMS-order is an extension of the \LMS-order defined in~\cite{NongZC2011}, to \LMS-prefixes. It can be proved that these definitions coincide for \LMS-substrings. 

\begin{definition}[\LMS-order]
Given two strings $S$ and $T$, let $U$ resp.\ $V$ be their \LMS-prefixes. We define $U <_{\LMS} V$ if either $V$ is a proper prefix of $U$, or neither is a proper prefix of the other and $U <_{\lex} V$. 
\end{definition}

The procedure {\tt Induced Sorting} for the conjugates of the multiset is analogous to the original one, except that strings are viewed cyclically. First, the array $A$ is subdivided into so-called {\em buckets}, one for each character. For $c\in \Sigma$, let $n_c$ denote the total number of occurrences of the character $c$ in the strings in ${\cal M}$. Then the buckets are  $[1,n_{a_1}],[n_{a_1}+1,n_{a_1}+n_{a_2}], \ldots,[N-n_{a_{\sigma}}+1,N]$, i.e., the $k$-th bucket will contain all conjugates starting with character $a_{k}$. The procedure {\tt Induced Sorting} first inserts all \LMS-positions at the end of their respective buckets, then induces the L-type positions in a left-to-right scan of $A$, and finally, induces the S-type positions in a right-to-left scan of $A$, possibly overwriting previously inserted positions. We need two pointers for each bucket ${\bf b}$, $\textit{head}({\bf b})$ and $\textit{tail}({\bf b})$, pointing to the current first resp.\ last free position of the bucket. 

\begin{quote}Procedure {\tt Induced Sorting}~\cite{NongZC2011} 
\begin{enumerate}
\item insert all \LMS-positions at the end of their respective buckets; initialize \textit{head}({\bf b}), \textit{tail}({\bf b}) to the first resp.\ last  position of the bucket, for all buckets ${\bf b}$
\item induce the L-type positions in a left-to-right scan of $A$: for $i$ from $1$ to $N-1$, if $A[i] = (j,d)$ then $A[\textit{head}(\textit{bucket}(T_d[j-1]))] \gets (j-1,d)$; increment $\textit{head}(\textit{bucket}(T_d[j-1]))$
\item induce the S-type positions in a right-to-left scan of $A$: for $i$ from $N$ to $2$, if $A[i] = (j,d)$ then $A[\textit{tail}(\textit{bucket}(T_d[j-1]))] \gets (j-1,d)$; decrement $\textit{tail}(\textit{bucket}(T_d[j-1]))$
\end{enumerate}
\end{quote}

At the end of this procedure, the \LMS-substrings are listed in correct relative \LMS-order (see Lemma~\ref{lemma:is_correct}), and they can be named according to their rank. For the recursive step, we define, for $i=1,\ldots,m$, a new string $T'_i$, where each \LMS-substring of $T_i$ is replaced by its rank. The algorithm is called recursively on ${\cal M'} = \{T'_1, \ldots, T'_m\}$ (Step 5). 

Finally (Step 6), the array $A'=\GCA({\cal M'})$ from the recursive step is mapped back into the original array, resulting in the placement of the \LMS-substrings in their correct relative order. This is then used to induce the full array $A$. All length-1 strings $T_i$ which were removed in Step 1 can now be inserted between the L- and S-type positions in their bucket (Lemma~\ref{lemma:is4omega}). See Figure~\ref{fig:ex1} for a full example. 

\begin{figure}
\resizebox{.8\textwidth}{!}{%
\begin{tikzpicture}
\setlength\tabcolsep{1.5pt} %
        \node (S1) [shape=rectangle,draw,rounded corners] {
            \begin{tabular}{rccccccccccccccccccccccccccc}
                 & & $T_1$  &   &   &   &   &  &  &  &  & $T_2$  &   &   &   &   &  &  &   &  &  &   &   &   & $T_3$ &  \\
                &  & 1  & 2  & 3  & 4  & 5  & 6  & 7 & 8 &  & 1  & 2  & 3  & 4  & 5  & 6  & 7 & 8 & 9 & 10 & 11  & 12  &   & 1 &  \\
                ${\cal M} =$ & $\{$ & G  & T  & A  & C  & A  & A  & C & G & , & C  & G  & G  & C  & A  & C & A & C  & A & C & G  & T  & ,  & C & $\}$ \\
                 &  & S  & L  & S  & L  & S  & S  & S & S &  & S  & L  & L  & L  & S  & L & S & L  & S & S & S  & L  &   &  &  \\
                 &  &   &   & $\ast$  &   & $\ast$  &   &  &  &  & $\ast$  &   &   &   & $\ast$  &  & $\ast$ &   & $\ast$ &  &   &   &   &  &  \\
            \end{tabular}
        };
        \node[align=center,anchor=south west] (lab2) at (S1.north west) {Step 2 - assign cyclic types to all positions of strings from ${\cal M}$};
        \node[align=center,anchor=south west] (lab1) at (lab2.north west) {Step 1 - remove strings of length 1 from ${\cal M}$};
        \node (S6) [below = of S1.south west, anchor=north west,shape=rectangle,draw,rounded corners] {
            \begin{tabular}{lccccccccccccccccccccccccccc}
                \phantom{${\cal M} =$} &\phantom{$\{$} & A  & \phantom{A}  & \phantom{A}  & \phantom{A}  & \phantom{A}  &\phantom{A}   & \multicolumn{1}{|c}{C} & \phantom{A} & \phantom{A} & \phantom{A} & \phantom{A} & \phantom{A}  & \phantom{A}  &  \phantom{A}  & \multicolumn{1}{|c}{G}  & \phantom{A} & \phantom{A}  & \phantom{A} & \phantom{A} & \multicolumn{1}{|c}{T}  & \phantom{A}  & \phantom{A} & \phantom{A} \\
                S$^\ast$ &  & & 5  & 7  & 9  & 3  & 5  & \multicolumn{1}{|c}{}  &  &  &  &   &   &   & 1  & \multicolumn{1}{|c}{}  &  &    &  &  &  \multicolumn{1}{|c}{} &   &   &  \\
                 &  & & 2  & 2  & 2  & 1  & 1  & \multicolumn{1}{|c}{}  &  &  &  &   &   &   & 2  & \multicolumn{1}{|c}{}  &  &  &   &   & \multicolumn{1}{|c}{}  &   &   &  \\
                 L & &  &   &   &   &   &   & \multicolumn{1}{|c}{4}  & 6 & 8 & 4 &   &   &   &   & \multicolumn{1}{|c}{3}  & 2 &  &   &  &  \multicolumn{1}{|c}{2} & 12  &   &  \\
                 $\longrightarrow$ & &  &   &   &   &   &   & \multicolumn{1}{|c}{2}  & 2 & 2 & 1 &   &   &   &   & \multicolumn{1}{|c}{2}  & 2 &  &   &  & \multicolumn{1}{|c}{1} & 2  &   &  \\
                 S & & 5 & 5 & 7 & 3 & 6 & 9 & \multicolumn{1}{|c}{}  &  &  &  &  & 1 & 7 & 10 & \multicolumn{1}{|c}{}  &  & 8 & 1 & 11 & \multicolumn{1}{|c}{} &   &  &   &  \\
                   $\longleftarrow$& & 1 & 2 & 2 & 1 & 1 & 2 & \multicolumn{1}{|c}{}  &  &  &  &  & 2 & 1 & 2  & \multicolumn{1}{|c}{}  &  & 1 & 1 & 2  & \multicolumn{1}{|c}{} &   &  &   &  \\
                  & & & & & & & & \\
               \phantom{${\cal M} =$} &\phantom{$\{$} & A  & \phantom{A}  & \phantom{A}  & \phantom{A}  & \phantom{A}  &\phantom{A}   & \multicolumn{1}{|c}{C} & \phantom{A} & \phantom{A} & \phantom{A} & \phantom{A} & \phantom{A}  & \phantom{A}  &  \phantom{A}  & \multicolumn{1}{|c}{G}  & \phantom{A} & \phantom{A}  & \phantom{A} & \phantom{A} & \multicolumn{1}{|c}{T}  & \phantom{A}  &  &  \\
               &  & 5 & 5  & 7  & 3  & 6  & 9  & \multicolumn{1}{|c}{4}  & 6 & 8  & 4 &  & 1 & 7 & 10  & \multicolumn{1}{|c}{3}  & 2 & 8 & 1 & 11 &  \multicolumn{1}{|c}{2} & 12  &   &  &  \\
                 &  & 1 & 2 & 2  & 1  & 1  & 2  & \multicolumn{1}{|c}{2}  & 2 & 2 & 1 &   & 2  & 1  & 2  & \multicolumn{1}{|c}{2}  & 2 & 1 &  1 &  2 & \multicolumn{1}{|c}{1}  & 2  &   &  &  \\ 
                 &  & $\ast$ & $\ast$ & $\ast$  & $\ast$  &   & $\ast$  & \multicolumn{1}{|c}{}  &  &  &  &   & $\ast$  &   &   & \multicolumn{1}{|c}{}  &  &  &  &   & \multicolumn{1}{|c}{}  &   &   &  &  \\
              \end{tabular}
        }; 
        \node[align=left,anchor=south west] (lab4) at (S6.north west) {Step 3 - use procedure {\tt Induced Sorting} to sort cyclic\\\LMS-substrings};
        \node (S3) [right = of S6.north east, anchor=north west, shape=rectangle,draw,rounded corners] {
            \begin{tabular}{rcccccccccccccc}
                 A  & A  & C & G & G & T & A  &    &    & a  \\
                A  & C  & A &   &   &   &    &    &    & b  \\
                A  & C  & G & T & C &   &    &    &    & c  \\
                 C  & G  & G & C & A &   &    &    &    & d  \\
            \end{tabular}
        }; 
        \node[align=left,anchor=south west] (lab3) at (S3.north west) {Step 4 - Assign names to\\cyclic $LMS$-substrings};
        \node (S3b) [below = of S3.south west, anchor=north west, shape=rectangle,draw,rounded corners] {
            \begin{tabular}{rcCCCCCC}
                 $T_1'$ & $=$ & b & a &   &   &    &   \\
                $T_2'$  & $=$ & d & b & b  & a  &    &    \\
            \end{tabular}
        }; 

        \node (S4) [below = 0.6cm of S6.south west, anchor=north west,,shape=rectangle,draw,rounded corners, ]{
            \begin{tabular}{rccccccccc}
                 & & $T_1'$  &   &   &  $T_2'$  &   &   &  \\
                 &  & 1  & 2  &  & 1  & 2  & 3  & 4  &  \\
                 ${\cal M}' =$ & $\{$ & b  & a  & , & d  & b  & b  & c  & $\}$ \\
                 &  & L  & S &  & L  & S & S & S  &  \\
                 &  &    & $\ast$  &  & & $\ast$  &   &  &   \\
            \end{tabular}\hspace{1.5cm}
            \begin{tabular}{lcc|ccc|c|ccc}
                \phantom{${\cal M} =$} &\phantom{$\{$} & a  & b & \phantom{A} & \phantom{A} & c  & d  &  \\
                S$^\ast$ &  & 2 &  &  & 2 &   &    &   &  \\
                 &  & 1 &  &  & 2 &   &   &   &  \\
                 L & &   & 1  &    &   &   & 1 &   &  \\
                 $\longrightarrow$  &\phantom{A} & \phantom{A}  & 1  & \phantom{A}   & \phantom{A}  & \phantom{A}  & 2 &   &  \\
                 S & & 2 & \phantom{A}   & 2  & 3 & 4 &   & \phantom{A}  &  \\
                  $\longleftarrow$ & & 1 &    & 2  & 2 & 2 &   &   &  \\
            \end{tabular}\hspace{1.2cm}
            \begin{tabular}{lcCCCCCCCC}
                 $A'$  &  & 2 & 1  & 2  & 3 & 4 & 1 &   &  \\
                   & & 1 & 1  & 2  & 2 & 2 & 2 &   &  \\
            \end{tabular}
        }; 
        \node[align=left,anchor=south west] (lab4) at (S4.north west) {Step 5 - recurse on new string multiset ${\cal M}'$};
        \node (S6b) [below = of S4.south west, anchor=north west,shape=rectangle,draw,rounded corners] {
            \begin{tabular}{lccccccc|cccccccc|ccccc|ccccccc}
                \phantom{${\cal M} =$} &\phantom{$\{$} & A  & \phantom{A}  & \phantom{A}  & \phantom{A}  & \phantom{A}  &\phantom{A}   & C & \phantom{A} & \phantom{A} & \phantom{A} & \phantom{A} & \phantom{A}  & \phantom{A}  &  \phantom{A}  & G  & \phantom{A} & \phantom{A}  & \phantom{A} & \phantom{A} & T  & \phantom{A}  & \phantom{A}  &  \phantom{A} \\
                S$^\ast$ &  & & 5  & 3  & 5  & 7  & 9  &   &  &  &  &   &   &   & 1  &   &  &    &  &  &   &   &   &  \\
                 &  & & 1  & 1  & 2  & 2  & 2  &   &  &  &  &   &   &   & 2  &   &  &  &   &   &   &   &   &  \\
                 L & &  &   &   &   &   &   & 4  & 4 & 6 & 8 &   &   &   &   & 3  & 2 &  &   &  &   2 & 12  &   &  \\
                 $\longrightarrow$ & &  &   &   &   &   &   & 1  & 2 & 2 & 2 &   &   &   &   & 2  & 2 &  &   &  &   1 & 2  &   &  \\
                 S & & 5 & 3 & 5 & 7 & 6 & 9 &   &  &  &  & \tikzmark{ul}1 & 1 & 7 & 10 &   &  & 8 & 1 & 11 &  &   &  &   &  \\
                 $\longleftarrow$ & & 1 & 1 & 2 & 2 & 1 & 2 &   &  &  &  & 3\tikzmark{br} & 2 & 1 & 2  &   &  & 1 & 1 & 2  &  &   &  &   &  \\
              \end{tabular}
        }; 
        \DrawBox[ultra thick, rounded corners,  black]{ul}{br}
        \tikz[overlay,remember picture]{
            \node[align=right,anchor=north east] (lab7) at ($(ul.north east)+(+15.75em,0.25em)$) {$T_3$};
            \draw[->] (lab7) -- ($(br.north east)+(+0.0em,+0.75em)$);
        }
        \node[align=left,anchor=south west] (lab4) at (S6b.north west) {Step 6 - use procedure {\tt Induced Sorting} to sort cyclic \LMS-substrings,\\ add length-1 strings in their respective positions};        
        \node (S7) [below = 0.6cm of S6b.south west , anchor=north west,shape=rectangle,draw,rounded corners]{
            \begin{tabular}{lcCCCCCCCCCCCCCcCCCCcCccc}
                 GCA& \phantom{$\{$} & 5 & 3  & 5  & 7  & 6 & 9  & 4  & 4  & 6 & 8  & {\bf 1} & {\bf 1} & 7 & 10 & 3  & 2  & 8 & {\bf 1}  & 11 &  2 & 12  &   &   \\
                 \phantom{${\cal M} =$}&  & 1 & 1 & 2 & 2 & 1 & 2 & 1 & 2  & 2  & 2  & {\bf 3}  & {\bf 2}  & 1  & 2  & 2  & 2 & 1 &  {\bf 1} &  2 & 1  & 2  &   &  \\ 
                 \eBWT\ &  & C & T & C & C & A & C & A & G  & A & A & C & T & A & A & G & C & C & G & C & G & G &   &  \\ 
            \end{tabular}
        }; 
        \node[align=left,anchor=south west] (lab4) at (S7.north west) {Generalized conjugate array of ${\cal M}$};
\end{tikzpicture}
}
\caption{The algorithm {\tt SAIS-for-\eBWT} on Example~\ref{ex:ex1}. Start positions of input strings are marked in bold.  
\label{fig:ex1}}
\end{figure}

\subsection{Correctness and running time}

The following lemma shows that the individual steps of {\tt Induced Sorting} are applicable for the $\omega$-order on conjugates of a multiset (part 1), that L-type conjugates (of all strings) come before the S-type conjugates within the same bucket (part 2), and that length-1 strings are placed between S-type and L-type conjugates (part 3). The second property was originally proved for the lexicographic order between suffixes in~\cite{KA2005}: 

\begin{lemma}[Induced sorting for multisets]\label{lemma:is4omega}
Let $U,V\in \Sigma^*$. 

\begin{enumerate}
\item If $U \prec_{\omega} V$, then for all $c\in \Sigma$, $cU \prec_{\omega} cV$. 
\item If $U[i] = V[j]$, $i$ is an L-type position, and $j$ an S-type position, then $conj_i(U) \prec_{\omega} conj_j(V)$. 
\item If $U[i] = V[j]=c$, $i$ is an L-type position, and $j$ an S-type position, then $conj_i(U) \prec_{\omega} c \prec_{\omega} conj_j(V)$. 
\end{enumerate}
\end{lemma}

\begin{proof}
{\em 1.} follows directly from the definition of $\omega$-order. {\em 3.} implies {\em 2.} For {\em 3.}, let $i'$ be the nearest character following $i$ in $U$ such that $U[i'] \neq c$. By Lemma~\ref{lemma:types}, $U[i']<c$, and thus $\conj_i(U) <_\lex c^{|U|}$, and therefore, $conj_i(U) \prec_{\omega} c$. Analogously, if $j'$ is the next character in $V$ s.t.\ $V[j'] \neq c$, then by Lemma~\ref{lemma:types}, $V[j']>c$, and therefore, $c \prec_{\omega} conj_j(V)$.  
\end{proof}

Next, we show that after applying procedure {\tt Induced Sorting}, the conjugates will appear in $A$ such that they are correctly sorted w.r.t.\ to the \LMS-order of their \LMS-prefixes, while the order in which conjugates with identical \LMS-prefixes appear in $A$ is determined by the input order of the \LMS-positions. 

\begin{lemma}[Extension of Thm.~3.12 of~\cite{NongZC2011}]\label{lemma:is_correct}
Let $T_1,T_2\in {\cal M}$, let $U$ be the \LMS-prefix of $\conj_i(T_1)$, with $i'$ the last position of $U$; let $V$ be the \LMS-prefix of $\conj_j(T_2)$, and $j'$ the last position of $V$. Let $k_1$ be the position of $\conj_i(T_1)$ in array $A$ after the procedure {\tt Induced Sorting}, and $k_2$ that of $\conj_j(T_2)$. 

\begin{enumerate}
    \item If $U<_{\LMS}V$, then $k_1<k_2$. 
    \item If $U=V$, then $k_1<k_2$ if and only if $\conj_{i'}(T_1)$ was placed before $\conj_{j'}(T_2)$ at the start of the procedure. 
\end{enumerate}
\end{lemma}

\begin{proof}
Both claims follow from Lemma~\ref{lemma:is4omega}, and the fact that from one \LMS-position to the previous one, there is exactly one run of L-type positions, preceded by one run of S-type positions. 
\end{proof}

The next lemma shows that the \LMS-order of the \LMS-prefixes respects the $\omega$-order. 

\begin{lemma}\label{lemma:lms_prefix}
Let $S,T\in \Sigma^*$, let $U$ be the \LMS-prefix of $S$ and $V$ the \LMS-prefix of $T$. If $U<_{\LMS}V$ then $S\prec_{\omega} T$.
\end{lemma}

\begin{proof}
If neither $U$ nor $V$ is a proper prefix one of the other, then there exists an index $i$ s.t.\ $S[i] = U[i] < V[i] = T[i]$, and therefore, $S \prec_{\omega} T$. Otherwise, $V$ is a proper prefix of $U$. Let $i=|V|$ and $c=V[i]$. Since both $U$ and $V$ are \LMS-prefixes, with $i$ being the last position of $V$ but not of $U$, this implies that $V[i] = T[i]$ is of type S, while $U[i]=S[i]$ is of type L. Let $j$ be the next character in $S$ s.t.\ $S[j] \neq c$, and $k$ be the next character in $T$ s.t.\ $T[k]\neq c$. By Lemma~\ref{lemma:types}, $S[j]<c$, $T[k]>c$, and by definition of $j,k$ all characters inbetween equal $c$. Then for $i'=\min(j,k)$, we have $S[i']<T[i']$, with $i'$ being the first position where $S$ and $T$ differ. Therefore, $S\prec_{\omega} T$. 
\end{proof}

\begin{theorem}
Algorithm {\tt SAIS-for-\eBWT} correctly computes the \GCA\ and \eBWT\ of a multiset of strings ${\cal M}$ in time $O(N)$, where $N$ is the total length of the strings in ${\cal M}$. 
\end{theorem}

\begin{proof}
By Lemma~\ref{lemma:types},  Step 2 correctly assigns the types. Step 3 correctly sorts the \LMS-substrings by Lemma~\ref{lemma:is_correct}. It follows from  Lemma~\ref{lemma:lms_prefix} that the order of the conjugates of the new strings $T'_i$ coincides with the relative order of the \LMS-conjugates. In Step 6, the \LMS-conjugates are placed in $A$ in correct relative order from the recursion; by Lemmas~\ref{lemma:is_correct} and~\ref{lemma:lms_prefix}, this results in the correct placement of all conjugates of strings of length $>1$, while the positioning of the length-1 strings is given by Lemma~\ref{lemma:is4omega}. 
 
 For the running time, note that Step 1 takes time at most $2N$. The {\tt Induced Sorting} procedure also runs in linear time $O(N)$. Finally, since no two \LMS-positions are consecutive, and we remove strings of length $1$, the problem size in the recursion step is reduced to at most $N/2$. 
\end{proof}

\subsection{Computing the BWT for one single string}

The special case where ${\cal M}$ consists of one single string leads to a new algorithm for computing the \BWT, since for a singleton set, the \eBWT\ coincides with the \BWT. To the best of our knowledge, this is the first linear-time algorithm for computing the \BWT\ {\em of a string without an end-of-string character} that uses neither Lyndon rotations nor end-of-string characters. 

We demonstrate the algorithm on a well-known example, $T = \textit{banana}$. We get the following types, from left to right: $LSLSLS$, and all three S-type positions are \LMS. We insert $2,4,6$ into the array $A$; after the left-to-right pass, indices are in the order $2,4,6,1,3,5$, and after the right-to-left pass, in the order $6,2,4,1,3,5$. The \LMS-substring \textit{aba} (pos.\ 6) gets the name $A$, and the \LMS-substring \textit{ana} (pos.\ 2,4) gets the name $B$. In the recursive step, the new string $T' = ABB$, with types $SLL$ and only one \LMS-position $1$, the \GCA\ gets induced in just one pass: $1,3,2$. This maps back to the original string: $6,2,4$, and one more pass over the array $A$ results in $6,4,2,1,5,3$ and the \BWT\ \textit{nnbaaa}. See Figure~\ref{fig:banana}.

\begin{figure}

\resizebox{1\textwidth}{!}{%
\begin{tikzpicture}
    \setlength\tabcolsep{1.5pt} %
    \node (S1) [shape=rectangle,draw,rounded corners] {
        \begin{tabular}{cccccc}
            1 & 2 & 3 & 4 & 5 & 6 \\
            b & a & n & a & n & a \\
            L & S & L & S & L & S \\
            & * & & * & & * \\
            & & & & & 
        \end{tabular}
    };
    \node[align=center,anchor=south west] (lab1) at (S1.north west) {Step 2};
    \node (S2) [right = 0.1cm of S1.north east, anchor=north west,shape=rectangle,draw,rounded corners] {
        \begin{tabular}{cCCCC|C|CC}
            & &$a$ & & & $b$ & $n$ & \\
            \hline 
            S$^\ast$& & 2 & 4 & 6 & & & \\
            L& & & & & 1 & 3 & 5 \\
            S& &6 & 2 & 4 & & & \\
            \hline 
            & &6 & 2 & 4 & 1 & 3 & 5\\
        \end{tabular}
        };
    \node[align=center,anchor=south west] (lab2) at (S2.north west) {Step 3};
    \node (S3) [right =  0.1cm of S2.north east, anchor=north west,shape=rectangle,draw,rounded corners] {
        \begin{tabular}{CCCCCCC}
            6 & & a & b & a & & A\\
            2 & & a & n & a & & B \\
            4 & & a & n & a & & B \\
            & & & & & & \\
            & & & & & & \\
        \end{tabular}
        };
    \node[align=center,anchor=south west] (lab2) at (S3.north west) {Step 4};
    \node (S4) [right =  0.1cm  of S3.north east, anchor=north west,shape=rectangle,draw,rounded corners] {
        \begin{tabular}{ccc}
            1 & 2 & 3  \\
            A & B & B \\
            S & L & L \\
            & * &  
        \end{tabular}\hspace{0.5cm}
        \begin{tabular}{c|cc}
            $A$ & $B$ & \\
            \hline 
            1 & &  \\
            & 3 & 2 \\
            \hline 
            1 & 3 & 2 \\
            & & 
        \end{tabular}
    };
    \node[align=center,anchor=south west] (lab4) at (S4.north west) {Step 5};
    \node (S5) [right =  0.1cm of S4.north east, anchor=north west,shape=rectangle,draw,rounded corners] {
        \begin{tabular}{lccc|c|cc}
            &$a$ & & & $b$ & $n$ & \\
            \hline 
            &6 & 4 & 2 & & & \\
            && & & 1 & 5 & 3\\
            \hline
            \GCA\ & 6 & 4 & 2 & {\bf 1} & 5 & 3\\
            \BWT\ & n & n & b & a & a & a\\
        \end{tabular}
        };
        \node[align=center,anchor=south west] (lab4) at (S5.north west) {Step 6};
    
\end{tikzpicture}
}
\caption{Example for computing the \BWT\ for one string\label{fig:banana}, start index marked in bold.}
\end{figure}

\section{\eBWT \ and prefix-free parsing}

In this section, we show how to extend the prefix-free parsing to build the $\eBWT$.  We define the \emph{cyclic prefix-free parse } for a multiset of strings ${\cal M} = \{T_1, T_2,\ldots,T_m\}$ (with $|T_i|=n_i$, $1\leq i\leq m$) as the multiset of parses ${\cal P} = \{P_1, P_2,\ldots,P_m\}$ with dictionary $D$, where we consider $T_i$ as circular, and $P_i$ is the parse of $T_i$.  We denote by $p_i$ the length of the parse $P_i$.

Next, given a positive integer $w$, let $E$ be a set of  strings of length $w$ called {\em trigger strings}. We assume that each string $T_h \in {\cal M}$ has length at least $w$ and at least one cyclic factor in $E$. %

We divide each string $T_h \in {\cal M}$ into overlapping phrases as follows: a phrase is a circular factor of $T_h$ of length $>w$ that starts and ends with a trigger string and has no internal occurrences of a trigger string. The set of phrases obtained from strings in ${\cal M}$ is the dictionary $D$. The parse $P_h$ can be computed from the string $T_h$ by replacing each occurrence of a phrase in $T_h$ with its lexicographic rank in $D$.

\begin{example}\label{ex:pfp}
Let 
${\cal M} = \{T_1:\textit{CACGTGCTAT},\, T_2:\textit{CCACTTGCTAGA},\, T_3:\textit{CACTTGCTAT}\}$ and let $E = \{{\it AC}, {\it G C}\}$. The dictionary $D$ of the multiset of parses ${\cal P}$ of ${\cal M}$ is $D = \{{\it ACCAC}, {\it ACGTGC}, {\it ACTTGC}, {\it GCTAGAC }, {\it GCTATCAC}\}$ and ${\cal P} = \{{\it 2\,5}, {\it 3\,4\,1}, {\it 3\,5}\}$, where $P_2 = {\it 2\, 5}$ means that the parsing of $T_2$ is given by the second and fifth phrases of the dictionary. Note that the string $T_2$ has a trigger string $\textit{AC}$ that spans the first position of $T_2$.
\end{example}

We denote by ${\cal S}$ the set of suffixes of $D$ having length greater than $w$. The first important property of the dictionary $D$ is that the set ${\cal S}$ {\em prefix-free}, i.e., no string in ${\cal S}$ is prefix of another string of ${\cal S}$.  This follows directly from \cite{DBLP:journals/almob/BoucherGKLMM19}.  %

\begin{example}
Continuing Example~\ref{ex:pfp}, we have that 
\begin{align*}
{\cal S} =\{ 
&{\it ACCAC}, {\it ACGTGC}, {\it ACTTGC}, {\it AGAC}, {\it ATCAC}, {\it CAC}, {\it CCAC}, {\it CGTGC},\\ &{\it CTAGAC},{\it CTATCAC}, {\it CTTGC}, {\it GAC}, {\it GCTAGAC}, {\it GCTATCAC}, {\it GTGC}, \\ &{\it TAGAC}, {\it TATCAC}, {\it TCAC}, {\it TGC}, {\it TTGC} \}
\end{align*}

\end{example}

The computation of $\eBWT$ from the prefix-free parse consists of three steps: computing the cyclic prefix-free parse of ${\cal M}$ (denoted as ${\cal P}$), computing the $\eBWT$ of ${\cal P}$ by using the algorithm described in Section \ref{sec:algo_ebwt}; and lastly, computing the $\eBWT$ of ${\cal M}$ from the $\eBWT$ of ${\cal P}$ using the lexicographically sorted dictionary $D=\{D_1, D_2, \ldots, D_{|D|}\}$ and its prefix-free suffix set ${\cal S}$.  We now describe the last step as follows. 
We define $\delta$ as the function that uniquely maps each character of $T_h[j]$ to the pair $(i,k)$, where with $1\leq i\leq p_h$, $k>w$, and $T_h[j]$ corresponds to the $k$-th character of the $P_h[i]$-th phrase of $D$. We call $i$ and $k$ the \emph{position} and the \emph{offset} of $T_h[j]$, respectively. Furthermore, we define $\alpha$ as  the function that uniquely associates to each conjugate $conj_j(T_h)$ the element $s\in \mathcal{S}$ such that $s$ is the $k$-th suffix of the $P_h[i]$-th element of $D$, where $(i,k)=\delta(T_h[j])$. By extension, $i$ and $k$ are also called the \emph{position} and the \emph{offset} of the suffix $\alpha(\conj_j(T_h))$.

\begin{example}\label{ex:pfp_2}
In Example~\ref{ex:pfp}, $\delta(T_2[4]) = (1,2)$ since $T_2[4]$ is the second character (offset 2) of the phrase ${\it ACTTGC}$, which is the first phrase (position 1) of $P_2$. Moreover, $\alpha(\conj_4(T_2))={\it CTTGC}$ since ${\it CTTGC}$ is the suffix of $D_3$, which is prefix of $\conj_4(T_2) = {\it CTTGCTAGACCA}$. 

\end{example}

\begin{lemma}\label{{lem:unique suffixes}}
Given two strings $T_g,T_h\in \mathcal{M}$, if $\alpha(conj_i(T_g))<_{\lex}\alpha(conj_i(T_h))$ it follows that  $conj_i(T_g)\prec_\omega conj_j(T_h)$.
\end{lemma}
\begin{proof}
It follows from the definition of $\alpha$ that $\alpha(conj_i(T_g))$ and $\alpha(conj_j(T_h))$ are prefixes of $conj_i(T_g)$ and  $conj_j(T_h)$, respectively.
\end{proof}

\begin{proposition}\label{prop:suffix order}
 Given two strings $T_g,T_h\in{\cal M}$. Let $\conj_i(T
 _g)$ and $\conj_j(T_h)$ be the $i$-th and $j$-th conjugates of $T_g$ and $T_h$, respectively, and let $(i',g')=\delta(T_g[i])$ and $(j',h')=\delta(T_h[j])$. 
 Then  $\conj_i(T_g) \prec_\omega \conj_j(T_h)$ if and only if either $\alpha(\conj_i(T_g)) <_\lex \alpha(\conj_j(T_h))$, or $\conj_{i'+1}(P_g) \prec_\omega \conj_{j'+1}(P_h)$, i.e., $P_g[i']$ precedes $P_h[j']$ in $\eBWT({\cal P})$.
  \end{proposition}
  \begin{proof}
   By definition of $\alpha$,  $\conj_i(T_g)=\alpha(\conj_i(T_g))T_g[i+g'']T_g[i+g''+1]\ldots T_g[i-1]$ and $\conj_j(T_h)=\alpha(\conj_j(T_h))T_h[j+h'']T_h[j+h''+1]\ldots T_h[j-1]$, where $g''=|\alpha(\conj_i(T_g))|$ and  $h''=|\alpha(\conj_j(T_h))|$, respectively.  Moreover, $\conj_i(T_g) \prec_\omega \conj_j(T_h)$ if and only if either $\alpha(\conj_j(T_h))<_\lex \alpha(\conj_j(T_h))$ or $\conj_{i+g''-w}(T_g) \prec_\omega \conj_{j+h''-w}(T_h)$, where $w$ is the length of trigger strings. It is easy to verify that the position of $T_g[i+g''-w]$ and $T_h[j+h''-w]$  is $i'+1$ and $j'+1$, respectively. Moreover, since $T_g[i+g''-w]$ and $T_h[j+h''-w]$  are the first character of a phrase, we have that $\conj_{i+g''-w}(T_g) \prec_\omega \conj_{j+h''-w}(T_h)$ if and only if $\conj_{i'+1}(P_g) \prec_\omega \conj_{j'+1}(P_h)$.
  \end{proof}

 Next, using Proposition~\ref{prop:suffix order}, we define how to build the $\eBWT$ of the multiset of strings ${\cal M}$ from ${\cal P}$ and $D$.  First, we note that we will iterate through all the suffixes in ${\cal S}$ in lexicographic order, and build the $\eBWT$ of ${\cal M}$ in blocks corresponding to the suffixes in ${\cal S}$.  Hence, it follows that we only need to describe how to build an $\eBWT$ block corresponding to a suffix $s\in{\cal S}$. Given $s \in {\cal S}$, we let ${\cal S}_s$ be the set of the lexicographic ranks of the phrases of $D$ that have $s$ as a suffix, i.e., ${\cal S}_s = \{ i \mid 1 \leq i \leq |D|, s \text{ is a suffix of } D_i\in D \}$. Moreover, given the string $T_h \in {\cal M}$, we let $\conj_i(T_h)$ be the $i$-th conjugate of $T_h$, let $j$ and $k$ be the position and offset of $T_h[i]$, and lastly, let $p$ be the position of $P_h[j]$ in $\eBWT({\cal P})$. We define $f(p,k)=D_{P_h[j]}[k-1]$ if $k>1$, otherwise $f(p,k)=D_{P_h[j-1]}[|D_{P_h[j-1]}|-w]$ where we view $P_h$ as a cyclic string. %

\begin{example}
In Example~\ref{ex:pfp}, %
$\eBWT({\cal P}) = {\it 4\,5\,1\,5\, 3\, 2\,3}$. 
Let us consider $\conj_4(T_2)$ and $\conj_3(T_3)$ that are both mapped to the suffix $CTT$ by the function $\alpha$.  By using Example \ref{ex:pfp_2}, the position and the offset of $T_2[4]$ are $1$ and $2$, respectively. 
The position of $P_2[1]=3$ in $\eBWT({\cal P})$ is $5$, because $\conj_2(P_2) \prec_\omega \conj_2(P_3)$. This implies that $\conj_4(T_2) \prec_\omega \conj_3(T_3)$ by Proposition~\ref{prop:suffix order}. 
Furthermore, $f(5,2) = T_2[3] = {\it A}$.
\end{example}

Finally, we let ${\cal O}_s$ be the set of pairs $(p,c)$ such that for all $d \in {\cal S}_s$, $p$ is the position of an occurrence of $d$ in $\eBWT({\cal P})$, and $c$ is the character resulting the application of the $f$ function considering as $k$ the offset of $s$ in $D_d$, i.e., $c = f(p,|D_d| - |s|+1)$. Formally, ${\cal O}_s = \{ (p, f(p,|D_{\eBWT({\cal P})[p]}| - |s|+1) \mid \eBWT({\cal P})[p] \in {\cal S}_s  \}$.

\begin{example}
In Example~\ref{ex:pfp}, if $s={\it CAC} \in {\cal S}$ and ${\cal S}_s = \{{\it 1}, {\it 5}\}$, where ${\it 1: ACCAC}$ and ${\it 5: GCTATCAC}$, then it follows that ${\cal O}_s = \{(3,{\it C}), (2,{\it T}), (4,{\it T}) \}$ since the phrase ${\it 1}$ is in position 3 in the $\eBWT({\cal P})$ and the suffix ${\it CAC}$ starts in position 3 of $D_1$, the character preceding the occurrences of ${\it CAC}$ corresponding to the phrase ${\it 1}$ is ${\it C}$. Analogously, the phrase ${\it 5}$ is in positions 2 and 4 in the $\eBWT({\cal P})$ and the suffix ${\it CAC}$ starts in position 6 of $D_5$, hence the character preceding the occurrences of ${\it CAC}$ corresponding to the phrase ${\it 5}$ is ${\it T}$.%
\end{example}

To build the $\eBWT$ block corresponding to $s \in {\cal S}$, we scan the set ${\cal O}_s$ in increasing order of the first element of the pair, i.e., the position of the occurrence in $\eBWT({\cal P})$, and concatenate the values of the second element of the pair, i.e., the character preceding the occurrence of $s$ in $T_h$.%
Note that if all the occurrences in ${\cal O}_s$ are preceded by the same character $c$, we do not need to iterate through all the occurrences but rather concatenate $|{\cal O}_s|$ copies of the character $c$.

\begin{example}
In Example~\ref{ex:pfp}, $\eBWT({\cal M})={\it GCCCTTT\underline{TCT}AAGGGAAATTTCCCCAATGTCC}$, where the block of the \eBWT\ corresponding to the suffix $s = {\it CAC}\in {\cal S}$ is underlined. Given ${\cal O}_s = \{(3,{\it C}), (2,{\it T}),(4,{\it T}) \}$, we generate the block by sorting ${\cal O}_s$ by the first element of each pair -- resulting in ${\cal O}_s = \{(2,{\it T}), (3,{\it C}),(4,{\it T})\}$ -- and   concatenating the second element of each pair obtaining ${\it TCT}$. 
\end{example}

\paragraph*{Keeping track of the first rotations.}
So far, we showed how to compute the first component of the \eBWT. Now we show how to compute the second component of the \eBWT\, i.e., the set of indices marking the first rotation of each string. The idea is to keep track of the starting positions of each text in the parse, by marking the offset of the first position of each string in the last phrase of the corresponding parse. We propagate this information during the computation of the \eBWT\ of the parse. When scanning the suffixes of ${\cal S}$, we check if one of the phrases sharing the same suffix $s\in{\cal S}$ is marked as a phrase containing a starting position, and if the offset of the starting position coincides with the offset of the suffix. If so, when generating the elements of ${\cal O}_s$, we mark the element corresponding to the occurrence of the first rotation of a string, and we output the index of the \eBWT\ when that element is processed.

\paragraph*{Implementation notes.}

In practice, as in~\cite{DBLP:journals/almob/BoucherGKLMM19}, we implicitly select the set of trigger strings $E$, by rolling a Karp-Rabin hash over consecutive windows of size $w$ and take as a trigger strings of length $w$ all windows such that their hash value is congruent $0$ modulo a parameter $p$. In our version of the PFP, we also need to ensure that there is at least one trigger string on each sequence of the collection. Hence, we change the way we select the trigger strings as follows. We define a set ${\cal D}$ of remainders and we select a window of length $w$ as a trigger string with hash value congruent $d$ modulo $p$ if $d \in {\cal D}$. Note that if we set ${\cal D} = \{0\}$ we obtain the same set of trigger strings as in the original definition. We choose the set ${\cal D}$ in a greedy way. We start with ${\cal D} = \{0\}$ by scanning the set of sequences and checking if the current sequence has a trigger string according to the current ${\cal D}$. As soon as we find one, we move to the next sequence. If we don't find any trigger string, we take the reminder of the last window we checked, and we include it in the set ${\cal D}$.

We note that we consider ${\cal S}$ to be the set of suffixes of the phrases of $D$ such that $s\in{\cal S}$ is not a phrase in $D$ nor it has length  smaller than $w$ in the implementation. This allows us to compute $f$ more efficiently since we can compute the preceding character of all the occurrences of a suffix in ${\cal S}$ from its corresponding phrase in $D$. 
Moreover, as in~\cite{DBLP:journals/almob/BoucherGKLMM19}, for each phrase in $D$, we keep an ordered list of their occurrences in the \eBWT{} of the parse. For a given suffix $s \in {\cal S}$, we do not generate ${\cal O}_s$ all at once and sort it -- but rather, we visit the elements of ${\cal O}_s$ in order using a min-heap as we merge the ordered lists of the occurrences in the \eBWT\ of the parse of the phrases that share the same suffix $s$.

\section{Experimental results}

We implemented the algorithm for building the $\eBWT$ and measured its performance on real biological data. We performed the experiments on a server with Intel(R) Xeon(R) CPU E5-2620 v4 @ 2.10GHz with 16 cores and 62 gigabytes of RAM running Ubuntu 16.04 (64bit, kernel 4.4.0). The compiler was \texttt{g++} version 9.4.0 with \texttt{-O3 -DNDEBUG -funroll-loops -msse4.2} options. We recorded the runtime and memory usage using the wall clock time, CPU time, and maximum resident set size from \texttt{/usr/bin/time}. The source code is available online at: \url{https://github.com/davidecenzato/PFP-eBWT}.

We compared our method (\ours) with the BCR algorithm implementation of~\cite{rope} (\rope), \gsufsort~\cite{louza2020gsufsort}, and \egap~\cite{egidi2019external}. We did not compare against \texttt{G2BWT}~\cite{diaz2021efficient}, \texttt{lba}~\cite{DBLP:journals/tcs/BonizzoniVPPR21}, and \texttt{BCR}~\cite{DBLP:journals/tcs/BauerCR13} since they are currently implemented only for short reads\footnote{\texttt{G2BWT} crashed and \texttt{BCR} did not terminate within 48 hours with the smallest of each dataset; \texttt{lba} works only with sequences of length up to 255}. %
We did not compare against \texttt{egsa}~\cite{egsa} since it is the predecessor of \egap{} or against methods that construct the \BWT\ of a multiset of strings using one of the methods we evaluated against, i.e.,  {\tt LiME}~\cite{meta}, {\tt BEETL}~\cite{DBLP:journals/bioinformatics/CoxBJR12}, {\tt metaBEETL}~\cite{Ander2013}, and  {\tt ebwt2snp}~\cite{PrezzaPSR19,PrezzaPSR20}. 

\subsection{Datasets}

We evaluated our method using 2,048 copies of human chromosomes 19 from the 1000 Genomes Project~\cite{1000genomes}; 10,000 {\it Salmonella} genomes taken from the GenomeTrakr project~\cite{genometrakr}, and 400,000 SARS-CoV2 genomes from EBI’s COVID-19 data portal~\cite{covid-data-portal}.  The sequence data for the {\it Salmonella} genomes were assembled, and the assembled sequences that had length less than 500 bp were removed.  In addition, we note that we replaced all degenerate bases in the SARS-CoV2 genomes with N's and filtered all sequences with more than 95\% N's. A brief description of the datasets is reported in Table~\ref{tab:realdatasets}. We used 12 sets of variants of human chromosome 19 (\chr19), containing $2^i$ variants for $i=0,\ldots, 11$ respectively. We used 6 collections of {\it Salmonella} genomes (\salmonella) containing 50, 100, 500, 1,000, 5,000, and 10,000 genomes respectively. We used 5 sets of SARS-CoV2 genomes (\sars) containing 25,000, 50,000, 100,000, 200,000, 400,000 genomes respectively. Each collection is a superset of the previous one. 

\begin{table*}[tbp]%
	 		\centering
 			\scalebox{1}{%
    	 		\sisetup{detect-weight = true,
    	 		detect-inline-weight = text,
    	 		table-number-alignment = right,
    	 		round-mode=places,
    	 		round-precision=2}
		 		\begin{tabular}{llS[table-format=3.0]S[table-format=6.2]S[table-format=4.2]}
		 			\hline
		 			Name & Description & $\sigma$ & {$n/10^6$} & {$n/r$}\\
		 			\hline
		 			{\chr19} & Human chromosome 19 & 5 & 121086.621263 &  2199.2135365161366 \\
		 			{\salmonella} & {\it Salmonella} genomes  & 4 & 48791.745168 & 112.71998894565677\\
		 			{\sars} & SARS-CoV2 genomes  & 5 & 11930.960556 &  1424.6511624241432 \\
		 			\hline
		 		\end{tabular}
		 	}
	 		\caption{Datasets used in the experiments.  We give the alphabet size in column 3.  We report the length of the file and the ratio of the length to the number of runs in the \eBWT\ in columns 4 and 5, respectively.
	 		\label{tab:realdatasets}}
	 	\end{table*}

\subsection{Setup}
We run \ours\ and \rope\ with 16 threads, and \gsufsort\ and \egap\ with a single thread since they do not support multi-threading.  Using \ours{}, we set $w=10$ and $p=100$. Furthermore, for \ours\ on the \salmonella\ dataset, we used up to three different remainders to build the \eBWT. We used \rope{} with the \texttt{-R} flag to exclude the reverse complement of the sequences from the computation of the $\BWT$.
All other methods were run with default parameters. 

We repeated each experiment five times, and report the average CPU time and peak memory for the set of chromosomes 19 up to 64 distinct variants, for {\it Salmonella} up to 1,000 sequences, and for all SARS-CoV2.
The experiments that exceeded 48 hours of wall clock time or exceeded 62 GB of memory were omitted for further consideration, e.g., 128 sequences of \chr19{}, 5000 sequences of \salmonella\, and 400,000 sequences of \sars\ for \egap{}. Furthermore, \gsufsort{} failed to successfully build the $\eBWT$ for 256 sequences of \chr19{}, 5000 sequences of \salmonella, and 400,000 sequences of \sars\ or more, because it exceeded the 62GB memory limit.

\subsection{Results}

In Figures~\ref{fig:chr19}, ~\ref{fig:salmonella}, and~\ref{fig:sars} we illustrate the construction time and memory usage to build the \eBWT\ and the \BWT\ of collections of strings for the chromosome 19 dataset, the {\it Salmonella} dataset, and the SARS-CoV2 dataset, respectively.

\ours{} was the fastest method to build the \eBWT\ of 4 or more sequences of chromosome 19, with a maximum speedup of 7.6x of wall-clock time and 2.9x of CPU time over \rope{} on 256 sequences of chromosomes 19, 2.7x of CPU time over \egap{} on 64 sequences, and 3.8x of CPU time over \gsufsort{} on 128 sequences. On {\it Salmonella} sequences, \ours{} was always the fastest method, except for 10,000 sequences where \rope{} was the fastest method on wall-clock time. \ours{} had a maximum speedup of 3.0x of wall-clock time over \rope{} on 100 sequences of \salmonella. Considering the CPU time, \ours\ was the fastest for $\geq$ 500 sequences with a maximum speedup of 1.7x over \rope\ on 100 sequences and 1.2x over \gsufsort{} and \egap\ on 1,000 sequences. On SARS-CoV2 sequences, \ours\ was always the fastest method, with a maximum speedup of 2.4x of wall-clock time over \rope\, while a maximum speedup of 1.3x of CPU time over \rope\ on 400,000 sequences, 2.9x over \gsufsort\, and 2.7x over \egap\ on 200,000 sequences of SARS-CoV2.

Considering the peak memory, on the chromosomes 19 dataset, \rope{} used the smallest amount of memory for 1, 2, 4, 8, and 2,048 sequences, while \ours{} used the smallest amount of memory in all other cases. \ours{} used a maximum of 5.6x less memory than \rope{} on 256 sequences of chromosomes 19, 28.0x less than \egap{} on 64 sequences, and 45.3x less than \gsufsort{} on 128 sequences. On {\it Salmonella} sequences, \ours{} used more memory than \rope{} for 50, 100, and 10,000 sequences, while \ours{} used the smallest amount of memory on all other cases. The largest gap between \rope{} and \ours{} memory peak is of 1.7x on 50 sequences. On the other hand, \ours{} used a maximum of 17.0x less memory than \egap{} and \gsufsort{} on 1,000 sequences. On SARS-CoV2 sequences, \ours\ always used the smallest amount of memory, with a maximum of 6.4x less memory than \rope\ on  25,000 sequences of SARS-CoV2, 57.1x over \gsufsort\ and \egap\ on 200,000 sequences.

The memory peak of \rope{} is given by the default buffer size of 10 GB, and the size of the run-length encoded $\BWT$ stored in the rope data structure. This explains the memory plateau on 10.5 GB of \rope{} on the chromosomes 19 dataset. However, \rope{} is able only to produce the $\BWT$ of the input sequence collection, while \ours{} can be trivially extended to produce also the samples of the conjugate array at the run boundaries with negligible additional costs in terms of time and peak memory. 
\begin{figure}[t!]
    \centering
 	\begin{subfigure}[c]{0.49\textwidth}
 	\centering
 		\includegraphics[width=\textwidth]{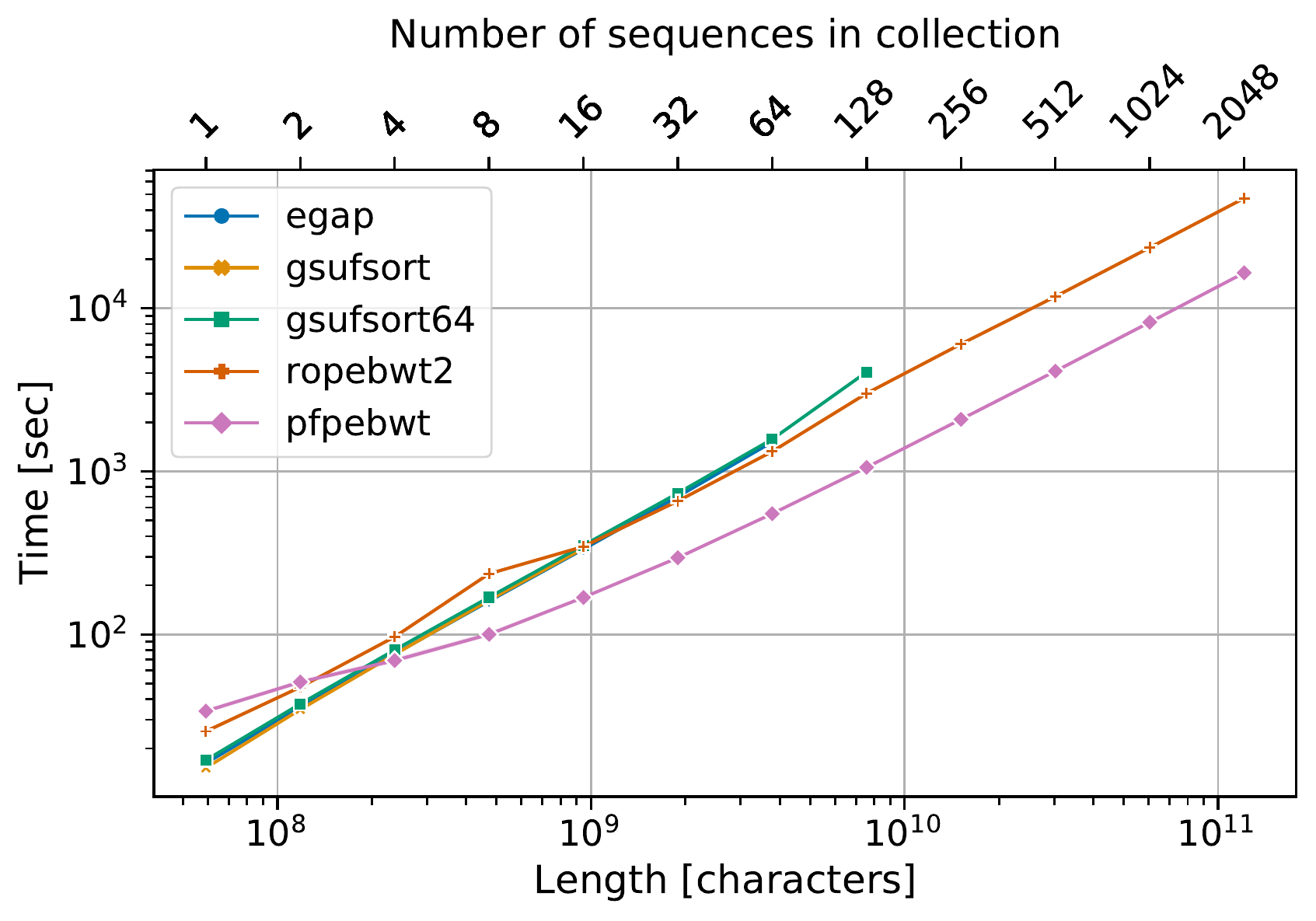}
 		\caption{Construction time. \label{fig:chr19:time}}
 \end{subfigure}\hfill	
 \begin{subfigure}[c]{0.49\textwidth}
 \centering
 		\includegraphics[width=\textwidth]{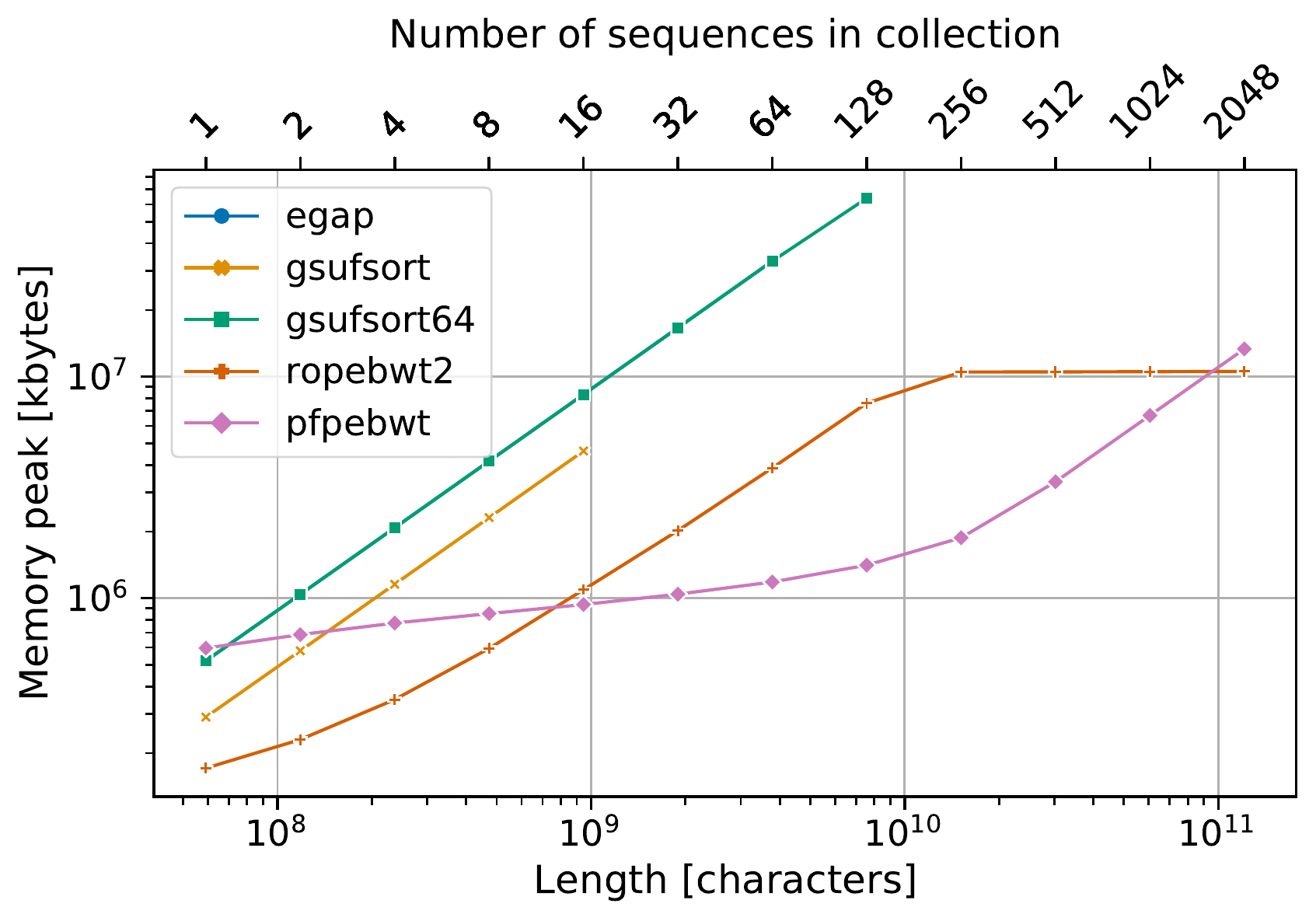}
 		\caption{Peak memory. \label{fig:chr19:space}}
\end{subfigure}
\caption{Chromosome 19 dataset construction CPU time and peak memory usage. We compare \ours{} with \rope, \gsufsort, and \egap. \label{fig:chr19}}
\end{figure}
\begin{figure}[t!]
    \centering
 	\begin{subfigure}[c]{0.49\textwidth}
 		\includegraphics[width=\textwidth]{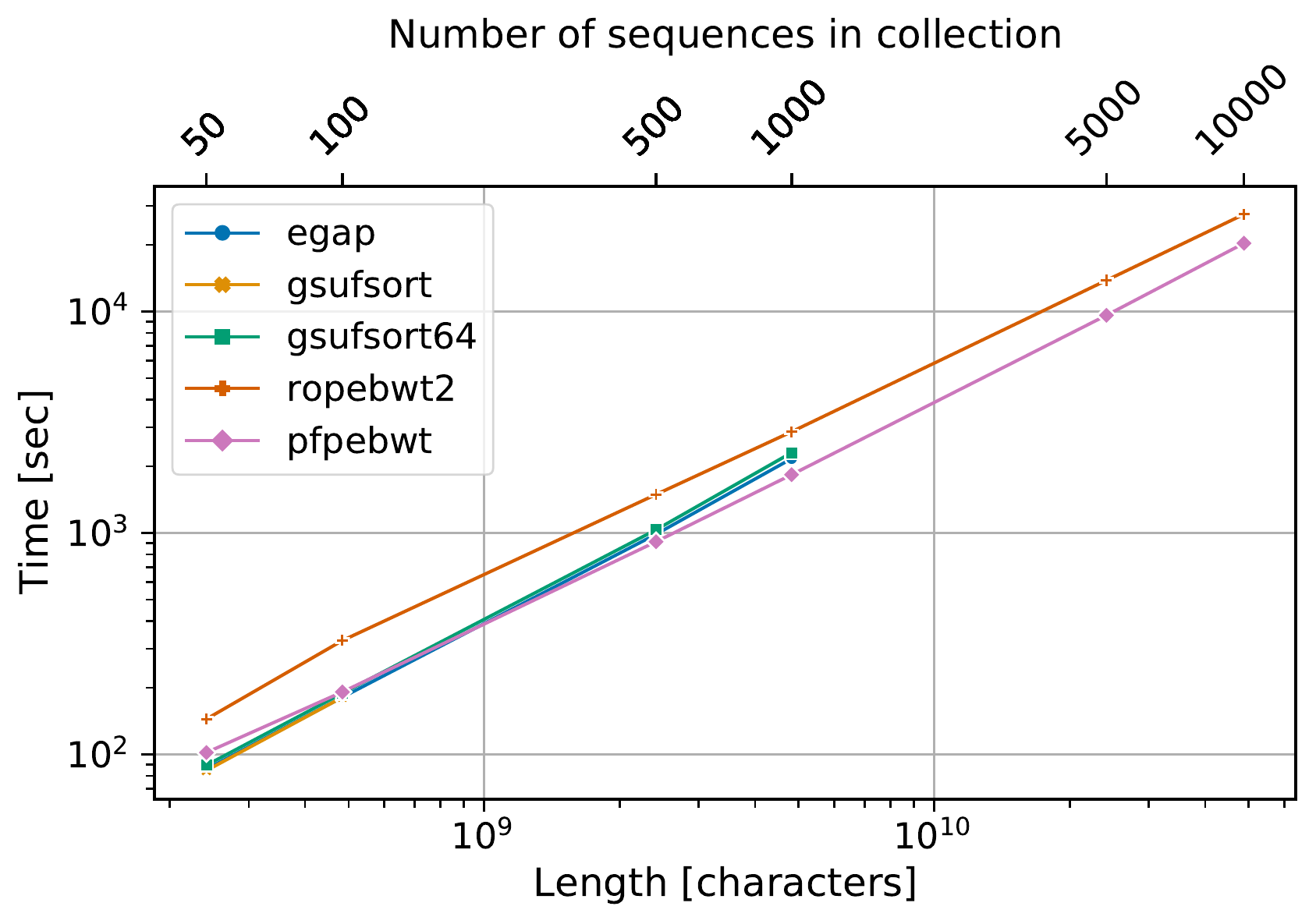}
 		\caption{Construction time. \label{fig:salmonella:time}}
 \end{subfigure}\hfill	
 \begin{subfigure}[c]{0.49\textwidth}
 \centering
 		\includegraphics[width=\textwidth]{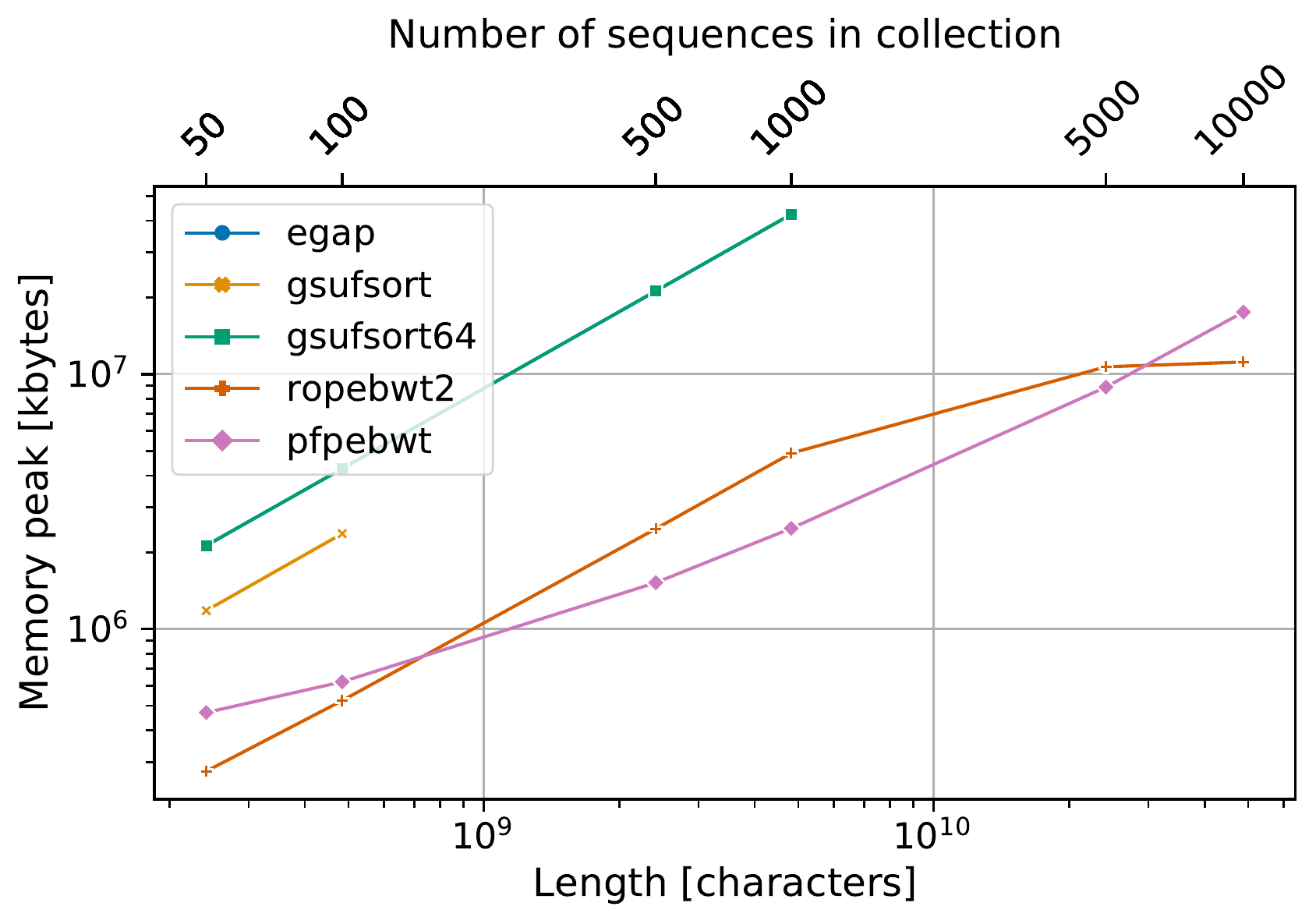}
 		\caption{Peak memory. \label{fig:salmonella:space}}
\end{subfigure}
\caption{{\it Salmonella} dataset construction CPU time and peak memory usage. We compare \ours{} with \rope, \gsufsort, and \egap. \label{fig:salmonella}}
\end{figure}
\begin{figure}[t!]
    \centering
 	\begin{subfigure}[c]{0.49\textwidth}
 		\includegraphics[width=\textwidth]{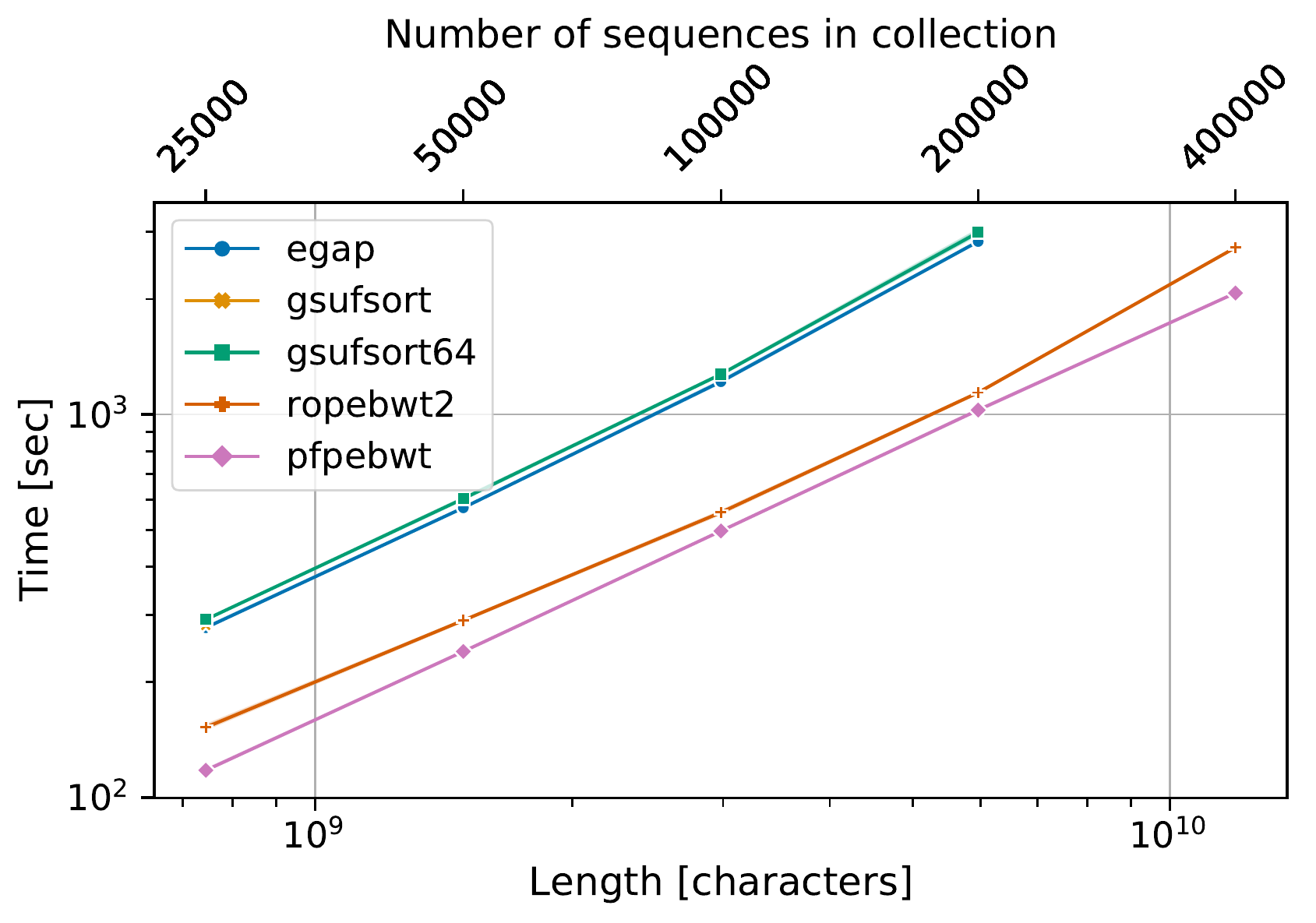}
 		\caption{Construction time. \label{fig:sars:time}}
 \end{subfigure}\hfill	
 \begin{subfigure}[c]{0.49\textwidth}
 \centering
 		\includegraphics[width=\textwidth]{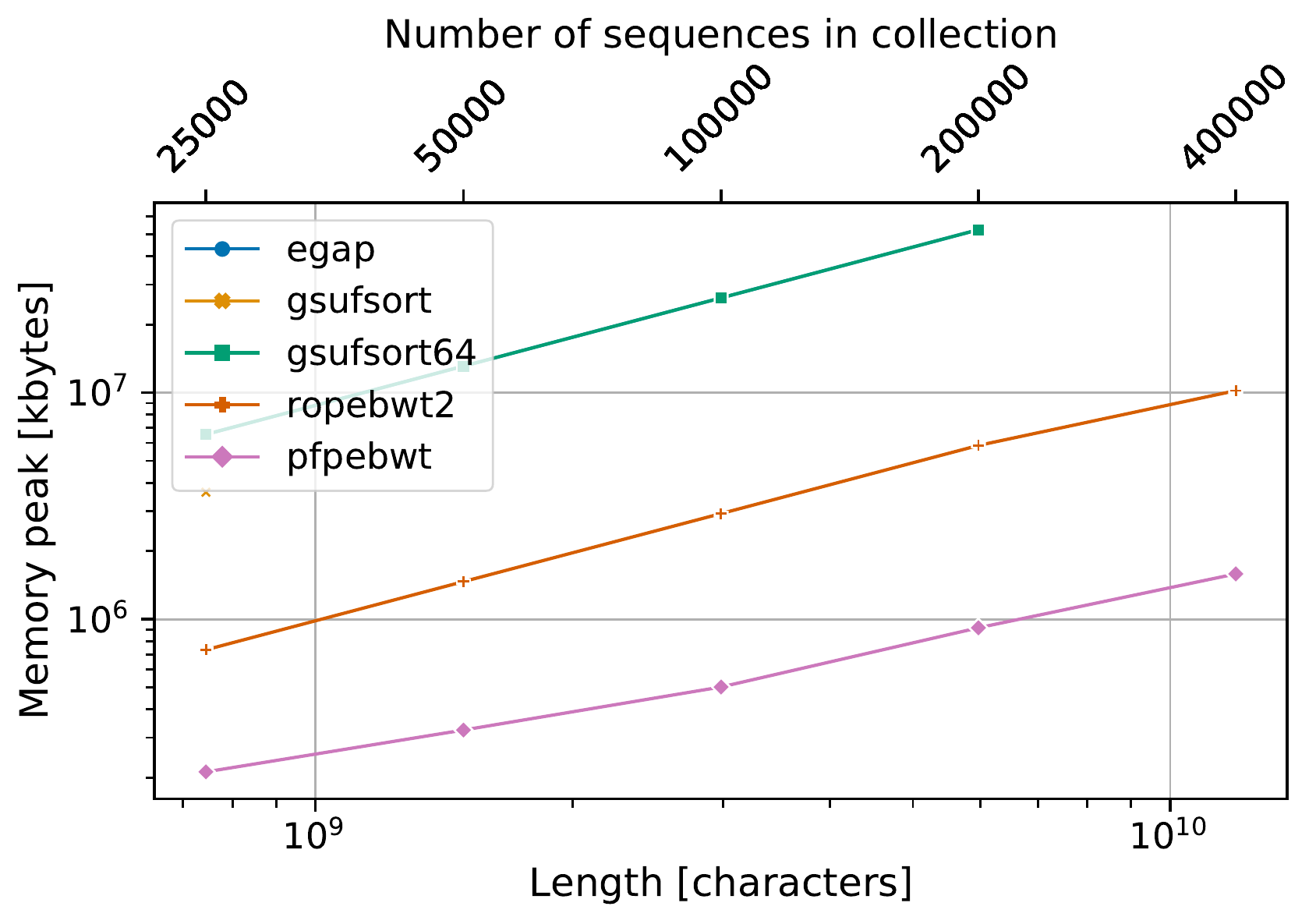}
 		\caption{Peak memory. \label{fig:sars:space}}
\end{subfigure}
\caption{SARS-CoV2 dataset construction CPU time and peak memory usage. We compare \ours{} with \rope, \gsufsort, and \egap. \label{fig:sars}}
\end{figure}

\section{Conclusion} 

We described the first linear-time algorithm for building the $\eBWT$ of a collection of strings that does not require the manipulation of the input sequence, i.e., neither the addition of an end-of-string character, nor computing and sorting the Lyndon rotations of the input strings. We also combined our algorithm with an extension of the prefix-free parsing to enable scalable construction of the $\eBWT$. We demonstrated $\ours$ was efficient with respect to both memory and time when the input is highly repetitive.  Lastly, we curated a novel dataset of 400,000 SARS-CoV2 genomes from EBI’s COVID-19 data portal, which we believe will be important for future benchmarking of data structures that have potential use in bioinformatics.  %

\bibliography{bibliography.bib}

\newpage

\appendix

\section{\eBWT\ missing examples}\label{app:exampleeBWT}
Full conjugate table for Example~\ref{ex:ex1}: $\cal M = \{ \textrm{GTACAACG,CGGCACACACGT,C}\}$. 

\medskip

\begin{figure}[h!]
\begin{center}
\begin{tabular}{lrrl}
            &  & \textrm{GCA} & $\preceq_\omega$-sorted conjugates~   \\
            &1 & (5,1)  & $\textrm{AACGGTA}\textbf{C}$     \\
             &2 & (3,1)  & $\textrm{ACAACGG}\textbf{T}$     \\
            &3 & (5,2)  & $\textrm{ACACACGTCGG}\textbf{C}$    \\
             &4 & (7,2)  &  $\textrm{ACACGTCGGCA}\textbf{C}$  \\
             &5 & (6,1)  & $\textrm{ACGGTAC}\textbf{A}$      \\
            &6 & (9,2)  & $\textrm{ACGTCGGCACA}\textbf{C}$  \\
            &7 & (4,1)  & $\textrm{CAACGGT}\textbf{A}$   \\
             &8 & (4,2)  & $\textrm{CACACACGTCG}\textbf{G}$   \\
              &9 & (6,2)  & $\textrm{CACACGTCGGC}\textbf{A}$    \\
            &10 & (8,2)  & $\textrm{CACGTCGGCAC}\textbf{A}$   \\
$\rightarrow$ &11 & (1,3)  & $\textbf{C}$                  \\
$\rightarrow$ &12 & (1,2)  & $\textrm{CGGCACACACG}\textbf{T}$  \\
            &13 & (7,1)  & $\textrm{CGGTACA}\textbf{A}$      \\
            &14 & (10,2)  & $\textrm{CGTCGGCACAC}\textbf{A}$   \\
            &15 & (3,2)  & $\textrm{GCACACACGTC}\textbf{G}$    \\
            &16 & (2,2)  & $\textrm{GGCACACACGT}\textbf{C}$    \\
            &17 & (8,1)  & $\textrm{GGTACAA}\textbf{C}$       \\
$\rightarrow$ &18 & (1,1)  & $\textrm{GTACAAC}\textbf{G}$    \\
            &19 & (11,2)  & $\textrm{GTCGGCACACA}\textbf{C}$    \\
            &20 & (2,1)  & $\textrm{TACAACG}\textbf{G}$    \\
             &21 & (12,2)  & $\textrm{TCGGCACACAC}\textbf{G}$   
\end{tabular}
\end{center}
\end{figure}

\begin{center}
\eBWT(\{\textrm{GTACAACG,CGGCACACACGT,C}\}) = \textrm{CTCCACAGAACTAAGCCGCGG}
\end{center}

\end{document}